\documentclass[12pt,psfig,a4]{article}
\makeatletter
\def\@seccntformat#1{\@ifundefined{#1@cntformat}%
   {\csname the#1\endcsname\quad}  
   {\csname #1@cntformat\endcsname}
}
\let\oldappendix\appendix 
\renewcommand\appendix{%
    \oldappendix
    \newcommand{\section@cntformat}{\appendixname~\thesection:\,\,}
}
\makeatother

\usepackage[T1]{fontenc}
\usepackage{amsfonts}
\usepackage[top=2.3cm,margin=1in]{geometry}
\usepackage{graphics}
\usepackage{lscape}
\usepackage{subfigure}
\usepackage{hyperref}
\hypersetup{
  colorlinks,
  citecolor=violet,
  linkcolor=red,
  urlcolor=blue}
\usepackage{graphicx}
\usepackage{setspace}
\usepackage{latexsym}
\usepackage{mathtools}
\usepackage{enumitem}
\usepackage{amsmath}
\usepackage{amsthm}
\usepackage{color,soul}
\usepackage{multirow}
\usepackage{hhline}
\usepackage[table]{xcolor}
\usepackage{booktabs}
\usepackage{rotating}
\usepackage[hang,small,bf]{caption}
\usepackage{caption}
\usepackage{adjustbox}
\usepackage{epstopdf}

\usepackage{epsfig}
\usepackage{dsfont}
\usepackage{imakeidx}
\usepackage[mathscr]{euscript}
\usepackage[overload]{empheq}
\usepackage{tikz}
\usepackage[ruled,lined,boxed]{algorithm2e}
\usepackage{epstopdf}
\usepackage[referable]{threeparttablex}

\usepackage[round]{natbib}

\DeclareMathOperator*{\argmax}{arg\,max}

\PassOptionsToPackage{
        natbib=true,
        style=authoryear-comp,
        hyperref=true,
        backend=biber,
        maxbibnames=99,
        firstinits=true,
        uniquename=init,
        maxcitenames=2,
        parentracker=true,
        url=false,
        doi=false,
        isbn=false,
        eprint=false,
        backref=true,
            }   {biblatex}

\newcommand{\Lyx}{L\kern-.1667em\lower.25em\hbox{y}\kern-.125emX\spacefactor1000}
\newcommand{\green}[1]{\textcolor{teal}{#1}}

\newtheorem{definition}{Definition}[section]
\newtheorem{theorem}[definition]{Theorem}
\newtheorem{proposition}[definition]{Proposition}
\newtheorem{lemma}[definition]{Lemma}
\newtheorem{corollary}[definition]{Corollary}
\newtheorem{remark}[definition]{Remark}%

\def\P{{\Bbb P}}
\def\F{{\cal F}}
\def\R{{\Bbb R}}
\def\E{{\Bbb E}}

\newcommand{\Var}{{\Bbb{V}\rm{ar}\,}}
\newcommand{\Cov}{{\Bbb{C}\rm{ov}}}
\newcommand{\tr}{{\rm tr}\;}

\input{tcilatex}

\begin{document}

\title{\textbf{Straightening skewed markets with an index tracking optimizationless portfolio}}
\author{Daniele Bufalo$^1$, Michele Bufalo$^2$, Francesco Cesarone$^3$$^\star$, Giuseppe Orlando$^4$ \\
{\small $^1$\emph{ University of Bari - Department of Computer Science}}\\
{\footnotesize daniele.bufalo@uniba.it}\\
{\small $^2$\emph{ Sapienza University of Rome - Department of Methods and Models for
Economics, Territory and Finance}}\\
{\footnotesize michele.bufalo@uniroma1.it}\\
{\small $^3$\emph{ Roma Tre University - Department of Business Studies}}\\
{\footnotesize francesco.cesarone@uniroma3.it}\\
{\small $^4$\emph{ University of Bari - Department of Finance and Economics}}\\
{\footnotesize giuseppe.orlando@uniba.it}\\
{\small $\star$ Corresponding author}\\
}
\date{\today}
\maketitle

\begin{abstract}
	\noindent
Among professionals and academics alike, it is well known that active portfolio management is unable to provide additional risk-adjusted returns relative to their benchmarks. For this reason, passive wealth management has emerged in recent decades to offer returns close to benchmarks at a lower cost. In this article, we first refine the existing results on the theoretical properties of oblique Brownian motion.
Then, assuming that the returns follow skew geometric Brownian motions and that they are correlated, we describe some statistical properties for the \emph{ex-post}, the \emph{ex-ante} tracking errors, and the forecasted tracking portfolio.
To this end, we develop an innovative statistical methodology, based on a benchmark-asset principal component factorization, to determine a tracking portfolio that replicates the performance of a benchmark by investing in a subset of the investable universe.
This strategy, named hybrid Principal Component Analysis (hPCA), is applied both on normal and skew distributions.
In the case of skew-normal returns, we propose a framework for calibrating the model parameters, based on the maximum likelihood estimation method. %
For testing and validation, we compare four alternative models for index tracking. The first two are based on the hPCA when returns are assumed to be normal or skew-normal. The third model adopts a standard optimization-based approach and the last one is used in the financial sector by some practitioners.
For validation and testing, we present a thorough comparison of these strategies on real-world data, both in terms of performance and computational efficiency.
A noticeable result is that, not only, the suggested lean PCA-based portfolio selection approach compares well versus cumbersome algorithms for optimization-based portfolios, but, also, it could provide a better service to the asset management industry.

\medskip
\noindent
\textbf{Keywords}: Index tracking, Passive fund management, Portfolio optimization, Tracking error, Skewed distributions.


\medskip
\noindent
\textbf{JEL classification}: G11, C44, C61, C53.
\end{abstract}

\section{Introduction}\label{sec:Intro}

Passive asset management has gained momentum in the last decade shifting the focus from fund picking and management selection to asset allocation.
This is because of clear evidence that actively managed funds consistently underperform their benchmarks.
This has been reported by popular publications such as the so-called S\&P Indices Versus Active (SPIVA) scorecard \citep{SPIVA} and the literature.
For example, \cite{Crane2018} argue ``stochastic dominance tests suggest no risk-averse investor should choose a random active fund over a random index fund''.

\noindent
To confirm that, Table \ref{T:SPIVA} reports that the percentage of USA equity funds outperformed by the S\&P500 index over 15 years of data varies from 63\% in the short term to 87\% in the long term (similar results are available for fixed income funds).
\begin{table}[!ht]

  \centering

\begin{tabular}{@{}llllll@{}}
\toprule
        & \multicolumn{5}{l}{Percentage of USA   Equity Funds Outperformed by Benchmarks} \\ \midrule
Index   & 1-YEAR (\%)   & 3-YEAR (\%)   & 5-YEAR (\%)   & 10-YEAR (\%)   & 15-YEAR (\%)   \\
S\&P500 & 63.17         & 71.24         & 77.97         & 82.06          & 86.92
 \\
 \bottomrule
\end{tabular}
\caption{Percentage of USA equity funds outperformed by the S\&P500 index \citep{SPIVA2020}.}  \label{T:SPIVA}
\end{table}

\noindent
A consequence of that is the growing importance of alternative investments, the offer of absolute return funds, and a massive boost to passive management so that, a company like Black Rock, with almost USD 9 trillion of assets under management \citep{BlackRock2021}, has become the biggest company in the wealth management industry.

Passive management managed funds could track their benchmark either by holding all the stocks or by investing in futures. The first approach is problematic in the sense that it may involve a large number of transactions, the second exposes the investor to the rollover and risk of derivative markets (mostly OTC) and related liquidity, credit and counterparty risk.

\noindent
Index tracking (IT) is a popular problem for tackling passive management portfolio constructions and it has received large attention in the literature.
Comprehensive surveys are contained in \cite{beasley2003evolutionary} and \cite{canakgoz2009mixed}.
Over the past decade, the inclusion of constraints, that better adapt the IT models to real-world applications, has led to an increase in the complexity of models.
Index tracking, basically, consists of a constrained optimization problem, where the distance between a given benchmark and the tracking portfolio is minimized by using a predefined number of assets less than those available in the investment universe. This cardinality constrained minimization with a quadratic tracking error measure has been proven to be NP-Hard by \cite{ruiz2009hybrid}.
Consequently, solving IT problems poses serious challenges in terms of the computational burden, especially for large-size problems.

%
%


The main aim of this paper is to develop a new statistical method based on an improved Principal Component Analysis (PCA), by which we provide an optimizationless strategy for building a tracking (small) portfolio that replicates a benchmark.
More precisely, assuming that assets' returns may be skew-normal, we present a procedure based on the similarities between the eigenvalues determined by a PCA on each benchmark-asset pair, for the selection of assets to be included in the tracking portfolio. Such a portfolio, as we shall see, exhibits promising index tracking capabilities. Apart from the dramatic gain in computation speed, the advantage is to provide a viable solution to the asset management industry in terms of both reliability of tracking error expectations and reductions of losses in turbulent markets.
To achieve this, our contributions are manifold.

First of all, a) we refine existing results on the theoretical properties of
the skew Brownian motion.
Then, b) assuming that the returns follow skew geometric Brownian motions and that they are correlated, we describe some statistical properties for the \emph{ex-post} tracking error, the \emph{ex-ante} tracking error, and the forecasted tracking portfolio.
To this end, c) we develop an innovative statistical methodology, based on a benchmark-asset principal component factorization, to determine a tracking portfolio that replicates the performance of a benchmark by investing in a subset of the investable universe.
This strategy, named hybrid PCA (hPCA), is applied both on normal and skew distributions.
In the case of skew-normal returns, d) we propose a procedure  for calibrating the model parameters, based on the maximum likelihood estimation method.
For testing, e) we compare four alternative models for index tracking. The first two are based on the hPCA when returns are assumed to be normal or skew-normal. The third model adopts a well-known approach from the literature  \citep[see, e.g.,][]{scozzari2013exact} and the last one is used in the financial sector by some practitioners.
For validation and testing, we present a thorough comparison of these strategies on real-world data, both in terms of performance and computational efficiency.

The rest of this article is organized as follow. Section \ref{Sec:LitRev} provides a brief account of the literature.
In Section \ref{Sec:MatMet}, we discuss several properties of the skew-normal random variables, and
of the skew arithmetic and geometric Brownian motions, which are relevant to our framework.
The portfolio selection models for index tracking are presented in Section \ref{Sec:algo}.
More precisely, in Section \ref{hPCA} we introduce the hPCA approach for normal and skew-normal returns, while in Section \ref{Sec:BaseMod} we describe two portfolio optimization models for index tracking used in the literature and the financial industry.
In Section \ref{Sec:NumSim}, all these IT strategies are tested on real-world data and compared both in
terms of performance and computational efficiency.
Section \ref{Sec:Conclusion} concludes.

\subsection{\textbf{Literature review}} \label{Sec:LitRev}

In this section, we provide a brief review that seeks to cover all the mathematical and statistical literature involved in the present research.
Clearly, the following account is not intended to be exhaustive of the literature, but only  highlights the different theoretical and practical aspects we have focused on.
\medskip

As mentioned in the introduction, in this article we illustrate an index tracking (IT) strategy that minimizes the tracking error between a tracking portfolio and its benchmark  by selecting a small number of constituents without resorting to any optimization method.
\paragraph{On index tracking portfolios}
\hfill
\break
\noindent
Index tracking portfolios have been studied in the literature since the 1990's \citep{Roll1992,Rudolf1999}.
Among the former, \cite{Frino2002} reported that actively managed equity funds in Australia were affected by a significant tracking error. To reduce that, \cite{Focardi2004} suggested clustering the portfolio in a way to discover the correlation and cointegration structure of the benchmark to avoid the burden of forecasting and optimization.

\noindent
Early index tracking overviews and related approaches can be found in \cite{beasley2003evolutionary,canakgoz2009mixed,guastaroba2012kernel,wang2012mixed,edirisinghe2013index,scozzari2013exact}.
One can mention  \cite{beasley2003evolutionary,guastaroba2012kernel} on heuristic approaches to the IT portfolio optimization and \cite{canakgoz2009mixed} on mixed-integer programming.
Differential Evolution and Combinatorial Search for Index Tracking (DECSI)
by \cite{Krink2009} proved to be an efficient and accurate heuristic approach for large problems. In case of small- and medium-sized problems, one can opt for the Mixed Integer Quadratic Programming (MIQP) approach provided by \cite{scozzari2013exact}.
\cite{Chen2012} proposed a binary algorithm to seek the maximization of the similarity between portfolio and reference index.
\cite{Bruni2015} presented an
indexing method to provide an excess return compared to the considered benchmark with the idea to maximize the average excess return and minimize underperformance. The advantage of this approach is that the replicating portfolio can be optimized using standard linear programming techniques \citep[see also][]{bruni2012new,bruni2017exact}.

\paragraph{On the tracking error for portfolio optimization}
\hfill
\break
\noindent
Asset managers often use tracking error (TE) as a measure of the risk of deviating from the benchmark, i.e., the TE risk.
Indeed, alternative approaches in portfolio optimization are those with TE constraints.
A common approach to TE-based portfolio optimization is to place some restrictions on it and to minimize or maximize other objectives. \cite{Jorion2003} was one of the first to suggest a risk/reward optimization based on a targeted TE. More recently,  \cite{Maxwell2018} has suggested an optimization model based on maximizing the Sharpe ratio on the constant TE frontier in the absolute risk/return space.
In any case, the problem with TE-bound portfolios is that they are ``bounded by an elliptical boundary in the mean/variance space and may not be efficient'' \citep{Maxwell2019}.

\medskip

The rationale behind our index tracking approach is to some extent to reduce the dimensionality of the investable universe into a small set of assets that better match the performance of the benchmark.
\paragraph{On PCA applied to portfolio management}
\hfill
\break
\noindent
Principal component analysis (PCA) is a technique that extracts some information from multivariate random variables data and transforms it into a new set of orthogonal variables, which are called principal components \citep[see, e.g.,][]{Abdi2010}.
When sorting the principal components from largest to smallest, the former retains most of the variability present in the original data.
The purpose of PCA is to reduce the dimensionality of a dataset, allowing to represent as much as possible the original variability of a multivariate random variable \citep[see, e.g.,][]{Jolliffe2003}.

The quality of the PCA method can be evaluated using cross-validation techniques, such as the bootstrap and the jackknife \citep[see, e.g.,][]{Choi2021}. Furthermore, PCA can be generalized both as correspondence analysis to handle qualitative variables and as multifactorial analysis \citep[see, e.g.,][]{kreinin1998principal,Murakami2020} to examine heterogeneous sets of variables.
Mathematically, PCA consists of the eigendecomposition of positive semi-definite matrices and the singular value decomposition (SVD) of rectangular matrices.
PCA has been applied to several problems in finance concerning risk decomposition \citep{Pasini2017}, portfolio optimization \citep{Zorivcic2020}, and stock trading \citep{Guo2020}.
\cite{Martellini2004}, through PCA, tried to reconcile investability with the representativeness in the hedge fund space.
Similarly, \cite{Nadkarni2018} used PCA and NeuroEvolution of Augmenting Topologies to derive a trading signal ``capable of high returns and daily profits with low associated risk''.
\cite{Antoniou2016} found an application of PCA to investor sentiment, beta and the cost of equity. \cite{Ouyang2019} suggests a deep autoencoder (i.e., a non-linear generalization of PCA) to track index performance and perform a dynamic weight calculation method. \cite{Cao2020} built a PCA-based stock price prediction model in conjunction with a three backpropagation neural network to devise an investment stock selection strategy.

\medskip
\paragraph{On skewed distributions of returns}
\hfill
\break
\noindent
As mentioned, we suggest a new statistical method based on an improved PCA, whereby we provide an optimizationless strategy for determining a tracking portfolio that replicates the benchmark by investing in a possibly small number of assets. This is considering  both normal and skew distributed returns.
Indeed, some scholars claim that standardized daily returns ``are approximately unconditionally normally distributed'' \citep{Andersen2001} or that ``they are IID Gaussian, with variance equal to 1'' \citep{Rogers2018}. In real life, instead,  returns, either standardized or not, do not conform to those hypotheses \citep[see, e.g.,][]{cont2001empirical,OB2021empirical}. To account for skewness and extra-kurtosis, one may consider skew-normal distributions as first introduced by \cite{Azzalini1985} and  \cite{Henze1986}.
\cite{Kim2001} was among those who exploited their properties, and who further improved the framework by introducing the t-skew distribution as scale mixtures of skew-normal distributions.   \cite{azzalini2003distributions} followed by providing a general account of the t-skew class of densities by stressing their ability to fit heavy-tailed and skewed data. In particular, the inclusion of the normal law and the shape parameter regulating the skewness, allows for a continuous variation from normality to non-normality \citep[see][]{azzalini2021overview}. Recently, \cite{Bufalo2022forecasting} described an application to forecasting portfolio returns with skew-geometric Brownian motions in presence of cross dependency between assets.

\noindent
Index tracking within the context of skewed distributions is the topic that we will introduce and discuss in the next section.




\section{Theoretical framework} \label{Sec:MatMet}

\subsection{\textbf{Preliminary concepts}}\label{Sec:Notation}

For completeness and readability, below we report some notions on skewed random variables (Section \ref{sec:Skew-normal}) and skewed stochastic processes (Section \ref{Sub1}).
Throughout the paper, we consider a filtered probability space $(\Omega, \P,\F,(\F_t)_{t\geq 0})$, where all random variables and stochastic processes are defined.

\subsubsection{Skew-normal density}\label{sec:Skew-normal}

We recall here some well-known notions about skew-normal random variables
\cite[see, e.g.,][]{azzalini2013skew}.
\begin{definition}[Skew-normal random variable]\label{sSN}
A random variable $X$ is said to be a standard skew-normal if its probability density function (pdf) is as follows
\begin{equation*}
f_{X}(x) = 2\phi(x)\Phi(\beta x) \qquad \mbox{with} \qquad x\in \R,
\end{equation*}
where $\phi$ is the standard normal pdf, $\Phi$ is the standard normal cumulative distribution function (cdf), and $\beta\in \R$ is the \emph{shape} parameter.
If $\beta=0$, then we have a normal distribution, while if the absolute value of $\beta$ increases, then the absolute value of the skewness increases.
More precisely, for $\beta >0$ ($\beta <0$) pdf is right (left) skewed.

\noindent
The standard skew-normal random variable can be generalized using the following affine transformation $Y=\xi+\omega X$, where $\xi\in \R$ is the location parameter and $\omega\in R_+$ is the scale parameter.
The pdf of a generalized skew-normal r.v. $Y\sim SN(\xi,\omega^2,\beta)$ is
\begin{equation*}\label{pSN}
f_{Y}(y)=\frac{2}{\omega}\phi\biggl( \frac{y-\xi}{\omega}\biggr)\Phi\biggl( \beta\,\frac{y-\xi}{\omega}\biggr) \qquad \mbox{with} \qquad y\in \R .
\end{equation*}
\end{definition}
\noindent In the following proposition we report the expressions for the moment-generating function, the expected value, and the variance of a skew-normal r.v..
\begin{lemma}[\cite{azzalini2013skew}, Lemma 2.2] \label{prop:moment-generating}
If $Y\sim SN(\xi,\omega^2,\beta)$ is a skew-normal random variable, its moment-generating function is given by
\begin{equation*}
\E\bigl[ e^{kY} \bigr]= 2e^{k\xi+\frac{k^2\omega^2}{2}}\Phi(k\delta \omega),
\end{equation*}
where
\begin{equation}\label{delta}
\delta=\frac{\beta}{\sqrt{1+\beta^2}}.
\end{equation}
Therefore, we have
\begin{equation*}\label{meanZ}
\E[Y] = \xi+\omega \delta \sqrt{\frac{2}{\pi}},
\end{equation*}
and
\begin{equation*}\label{varZ}
\Var[Y]=\biggl( 1-\frac{2\delta^2}{\pi}\biggr)\omega^2.
\end{equation*}
\end{lemma}
\noindent From Lemma \ref{prop:moment-generating}, it can be proved an interesting property that will be helpful in the following.
\begin{proposition}[\cite{azzalini2013skew}, Proposition 2.3] \label{prop21}

If $Y_1\sim SN(\xi,\omega^2,\beta)$ and $Y_2\sim N(\mu,\sigma^2)$ are two independent random variables, then
\begin{equation*}
Y_1+Y_2\sim SN(\xi+\mu,\omega^2+\sigma^2,\widetilde{\beta}),
\end{equation*}
where
$$
\widetilde{\beta}=\frac{\beta}{\sqrt{1+(1+\beta^2)\frac{\sigma^2}{\omega^2}}}.
$$
\end{proposition}
%

\noindent
In the next section, we provide some concepts about the skew arithmetic and geometric Brownian motions by which asset returns and prices are modeled, respectively, in our approach to IT.

\subsubsection{Skew arithmetic and geometric Brownian motion}\label{Sub1}

The Skew Arithmetic Brownian Motion (SABM) was firstly proposed by \cite{Ito1965} as a  generalization of the classical (Arithmetic) Brownian motion.
%
%

\noindent
In the following, we recall a standard definition of SABM process
\citep[see][]{atar2015multi,azzalini2021overview,Ito1974}.

\begin{definition}[Skew Arithmetic Brownian Motion]\label{SBM}
The SABM $Y(t)$ with $t \in [0, T]$ is a continuous-time stochastic process characterized by the following properties:
\begin{itemize}
\item[i)] $Y(0)=0$;
\item[ii)] for any $t \in [0, T]$ has continuous sample paths;
\item[iii)] for any $t_1,t_2 \in [0, T]$, with $t_1<t_2$, the increments $Y(t_2)-Y(t_1)$ are independent and skew-normally distributed, i.e., $Y(t_2)-Y(t_1)\sim SN(0,t_2-t_1,\beta)$.
\end{itemize}
\end{definition}
\noindent As shown by \cite{corns2007skew} (see Proposition 2.1),
a SABM $Y(t)$ can be constructed by the sum of a Brownian motion and a reflected Brownian motion, namely
\begin{equation}\label{sbm_alternative}
Y(t)=\sqrt{1-\delta^2} \, W_1(t)+\delta  \,  \lvert W_2(t) \rvert ,
\end{equation}
where $\delta$ is as in \eqref{delta}, and $W_1(t)$ and $W_2(t)$ are independent Brownian motions.
Hereafter, we represent a SABM $Y(t)$ by Eq. \eqref{sbm_alternative}, which complies with Definition \ref{SBM}.
\begin{remark} \label{rem:Zhu}
As described by \cite{zhu2018new}, the reflected Brownian motion does not have stationary increments, and the same occurs for SABM.
Then, for any $0\leq s < t$, $Y(t)$ has the following conditioning pdf
\begin{equation*}\label{cond_dist_Y}
f_{Y(t) \mid Y(s)}=f_{\sqrt{1-\delta^2}  \, W_1(t) \mid \sqrt{1-\delta^2}  \, W_1(s)} \ast f_{\delta  \, \lvert W_2(t) \rvert \delta  \, \lvert W_2(t) \rvert }
\end{equation*}
where $\ast$ denotes the convolution product,
\begin{equation}\label{f1}
f_{\sqrt{1-\delta^2}\cdot W_1(t) \mid \sqrt{1-\delta^2}\cdot W_1(s)}(x_1 \mid u_1(s))=\frac{1}{\sqrt{2\pi(1-\delta^2)(t-s)}}e^{-\frac{(x_1-u_1(s))^2}{2(1-\delta^2)(t-s)}} \quad (x_1\in \R),
\end{equation}
and
\begin{equation}\label{f2}
f_{\delta  \mid W_2(t) \mid \delta  \mid W_2(t) \mid }(x_2 \mid u_2(s))=\frac{1}{\delta\sqrt{2\pi (t-s)}}\biggl(e^{-\frac{(x_2-u_2(s))^2}{2\delta^2(t-s)}}+e^{-\frac{(x_2+u_2(s))^2}{2\delta^2(t-s)}}\biggr) \quad (x_2\in \R_+).
\end{equation}
In Expressions \eqref{f1} and \eqref{f2},
$u_1(s)$ and $u_2(s)$ represent realizations of $\sqrt{1-\delta^2}\cdot W_1(s)$ and $\delta \cdot  \lvert W_2(s) \rvert $, respectively.
\end{remark}
Using the results presented in Remark \ref{rem:Zhu}, in the following proposition
we show how to compute the conditional expectation and variance of a SABM.
\begin{theorem}[Conditional expectation and variance of a SABM] \label{prop22}
Given a SABM $Y(t)$, defined as in Eq. \eqref{sbm_alternative}, for any $0\leq s <t$, we have that
\begingroup\makeatletter\def\f@size{11}\check@mathfonts
\begin{equation}\label{EYts}
\E[Y(t) \mid \F_s]=u_1(s)-u_2(s)+2u_2(s)\Phi\biggl( \frac{u_2(s)}{\delta\sqrt{t-s}}\biggr)+\delta\sqrt{\frac{2(t-s)}{\pi}}\cdot\phi\biggl( \frac{u_2(s)}{\delta\sqrt{t-s}}\biggr),
\end{equation}
\endgroup
and
%
%
\begingroup\makeatletter\def\f@size{10}\check@mathfonts
\begin{equation}\label{var_oldts}
\Var\![]Y(t) \mid \F_s]=(t-s)+u_2^2(s)-\biggl( u_2(s)\biggl(2\Phi\biggl( \frac{u_2(s)}{\delta\sqrt{t-s}}\biggr)-1\biggr)+\delta\sqrt{\frac{2(t-s)}{\pi}}\cdot\phi\biggl( \frac{u_2(s)}{\delta\sqrt{t-s}}\biggr)\biggr)^2,
\end{equation}
\endgroup
%
%
where $u_1(s)$ and $u_2(s)$ are realizations of $\sqrt{1-\delta^2} \, W_1(s)$ and $\delta \,  \lvert W_2(s) \rvert$, respectively. 
\end{theorem}
\begin{proof}
Due to the independence of $W_1(t)$ and $W_2(t)$, we have
$$
\E[Y(t) \mid \F_s]= u_1(s) +\E[\delta \,  \lvert W_2(t) \rvert  \mid \F_s],
$$
with
\begin{equation}\label{cond_exp_abs_W2}
\E[\delta \, \lvert W_2(t) \rvert  \mid \F_s]=\int_0^{+\infty} x_2\, f_{\delta  \lvert W_2(t) \rvert \mid \delta  \lvert W_2(t) \rvert }(x_2 \mid u_2(s))\,dx_2,
\end{equation}
where the conditional pdf $f_{\delta  \lvert W_2(t) \rvert \mid \delta  \lvert W_2(t) \rvert }$ is defined as in Eq. \eqref{f2}, and $u_1(s)$ and $u_2(s)$ represent realizations of $\sqrt{1-\delta^2} \, W_1(s)$ and $\delta \,  \lvert W_2(s) \rvert $, respectively.
Solving the integral of Expression \eqref{cond_exp_abs_W2}, we have
\begingroup\makeatletter\def\f@size{12}\check@mathfonts
$$
\biggl[ \frac{2\sqrt{\frac{\pi}{2}}u_2(s)\bigl( \Phi\bigl( \frac{x_2-u_2(s)}{\delta \sqrt{t-s}} \bigr)- \Phi\bigl( \frac{x_2+u_2(s)}{\delta \sqrt{t-s}} \bigr)\bigr) -\delta\sqrt{t-s}\cdot\bigl( \phi\bigl( \frac{x_2-u_2(s)}{\delta\sqrt{t-s}} \bigr)+\phi\bigl( \frac{x_2+u_2(s)}{\delta\sqrt{t-s}} \bigr) \bigr) }{ \sqrt{2\pi}} \biggr]_0^{+\infty} \, ,
$$
\endgroup
namely
\begin{equation}\label{cond_exp_abs_W}
\E[\delta \, \lvert W_2(t) \rvert  \mid \F_s] = u_2(s)\biggl(2\Phi\biggl( \frac{u_2(s)}{\delta\sqrt{t-s}}\biggr)-1\biggr)+\delta\sqrt{\frac{2(t-s)}{\pi}}\cdot\phi\biggl( \frac{u_2(s)}{\delta\sqrt{t-s}}\biggr),
\end{equation}
which leads to Eq. \eqref{EYts}.
Expression \eqref{var_oldts} can be obtained knowing that
$$
\Var\![\delta \lvert W_{2}(t) \rvert  \mid \F_s]=\E[\delta^2W_{2}^2(t) \mid \F_s]-\bigl( \E[\delta \mid W_{2}(t) \mid  \mid \F_s]\bigr)^2,
$$
where $\E[\delta^2W_{2}^2(t) \mid \F_s]=\delta^2(W^2_{2}(s)+(t-s))$ and  $\bigl(\E[\delta \lvert W_{2}(t) \rvert  \mid \F_s]\bigr)^2$ is given by the square of Eq. \eqref{cond_exp_abs_W}.
\end{proof}
\smallskip

\noindent
Now, we introduce the so-called Skew Geometric Brownian Motion (SGBM) that can be seen as a generalization of the classical Geometric Brownian motion. The SGBM will be used to model the benchmark and the asset prices of the investment universe considered.
\begin{definition}[Skew Geometric Brownian Motion]\label{sGBM}
A stochastic process $S(t)$ is said to be a Skew Arithmetic Brownian Motion (SGBM) if for any $0\leq s<t$ it has the following representation
%
\begin{equation*}\label{SDE_sgbm}
S(t)=S(s) \exp\biggl(\mu  (t-s)+\sigma (Y(t)-Y(s)) \biggr) \qquad S(s)>0,
\end{equation*}
where $Y(\cdot)$ is a Skew Arithmetic Brownian Motion, $\mu\in \R$ and $\sigma \in \R_+$.
%
\end{definition}

\noindent
In the next section, exploiting the concepts described above, we will define the tracking portfolio and then the tracking error measure between this portfolio and the benchmark on a skew distributed market.

\subsection{\textbf{Tracking portfolio}}\label{sub:tp}

We assume that the benchmark price $S_B$ and the asset prices $S_i$, with $i\in N= \{1, \ldots, n\}$, of an investment universe follow SGBMs.
%
%
Therefore, from Definition \ref{sGBM}, for any $t> 0$, $S_B$ and $S_i$ are defined by the following dynamics
\begin{equation}\label{I_Si}
\begin{cases}
S_B(t)=S_B(0) \exp \bigl( \mu_B t+\sigma_B Y_B(t) \bigr) \\
S_{i}(t)=S_{i}(0) \exp \bigl( \mu_{i}t+\sigma_{i}Y_{i}(t) \bigr),
\end{cases}
\end{equation}
where $Y_B(t)$ is a Skew (Arithmetic) Brownian Motion with shape parameter $\beta_B$, which is defined as in Eq. \eqref{sbm_alternative}
\begin{equation*}\label{Y0}
Y_B(t)=\sqrt{1-\delta^2_B} \, W_{1}^{B}(t)+\delta_B \lvert W_{2}^{B}(t) \rvert ,
\end{equation*}
where $\delta_B=\frac{\beta_B}{\sqrt{1+\beta^2_B}}$ as in \eqref{delta},
and $W_{1}^{B}(t)$ and $W_{2}^{B}(t)$ are independent Brownian motions.
Furthermore, we assume that
\begin{equation}\label{Yi1}
Y_{i}(t)=\rho_i Y_B(t)+\sqrt{\biggl( 1-\frac{2 \delta_B^2}{\pi}\biggr)(1-\rho_i^2)} \, W_i(t),
\end{equation}
where $W_i(t)$ is a Brownian motion that is statistically independent from $Y_B(t)$,
and $\rho_i\in (-1,1)$.

\noindent
The stochastic processes $Y_i$ are still SGBMs, i.e., they verify the properties of Definition \ref{SBM}, as shown by the following proposition.
\begin{proposition} \label{prop23}
For any $i \in N$ and $t> 0$, the stochastic processes $Y_{i}$ defined in Eq. \eqref{Yi1} are Skew Arithmetic Brownian Motions with shape parameter
\begin{equation}\label{beta1}
\beta_i=\frac{\beta_B}{\sqrt{1+(1+\beta_B^2)\bigl( 1-\frac{2\delta_B^2}{\pi}\bigr)\bigl(\frac{1}{\rho_i^2}-1\bigr)}},
\end{equation}
where $\rho_i=Corr(Y_B,Y_{i})$ represents the (linear) correlation between $Y_B$ and $Y_{i}$.
\end{proposition}
\begin{proof}
Using Proposition \ref{prop21}, it is straightforward to see that the process $Y_{i}$ is still a Skew Arithmetic Brownian Motion with variance
\begin{equation*}\label{VarYi}
\Var[Y_{i}(t)]=\rho_i^2 \Var[Y_B(t)]+\biggl( 1-\frac{2 \delta_B^2}{\pi}\biggr)(1-\rho_i^2)\Var[W_i(t)]=\biggl( 1-\frac{2 \delta_B^2}{\pi}\biggr) t \, ,
\end{equation*}
and shape parameter $\beta_i$ as in Eq. \eqref{beta1}.

\noindent
Furthermore, due to the independence between $Y_B(t)$ and $W_i(t)$,
we have
\begin{eqnarray}
  \Cov(Y_{B}(t),Y_{i}(t)) &=& \E[Y_{B}(t) Y_{i}(t)]-\E[Y_B(t)]\E[Y_i(t)] \nonumber \\
   &=& \rho_i \E[Y^2_B(t)]+ \sqrt{\biggl( 1-\frac{2 \delta_B^2}{\pi}\biggr)(1-\rho_i^2)} \E[Y_{B}(t)]  \E[W_{i}(t)]- \rho_i \bigl(\E[Y_B(t)]\bigr)^2 \nonumber \\
   &=& \rho_i \Var[Y_B(t)]=\rho_i\biggl( 1-\frac{2\delta_B^2}{\pi}\biggr) t \, . \nonumber
\end{eqnarray}
Hence, this implies that
$$
Corr(Y_B,Y_{i})=\frac{\Cov(Y_B,Y_i)}{\sqrt{\Var[Y_B]\cdot \Var[Y_i]}}= \rho_i.
$$
\end{proof}

\medskip

\noindent
Using Expressions \eqref{I_Si} for
the benchmark price $S_B$ and the asset prices $S_i$ with $i\in N$,
we can define their log-returns from time $s$ to $t$ ($s \rightarrow t$)
as follows
%
%
\begin{eqnarray}
R_B^{s \rightarrow t} &=& \ln\biggl( \frac{S_B(t)}{S_B(s)}\biggr)=\mu_B (t-s) +\sigma_B  (Y_B(t)-Y_B(s)) \nonumber \\
 R_i^{s \rightarrow t}  &=& \ln\biggl( \frac{S_i(t)}{S_{i}(s)} \biggr)=\mu_i (t-s) + \sigma_i  (Y_i(t)-Y_i(s)) \, , \nonumber
\end{eqnarray}
%
respectively.
%
Furthermore, denoting by $\Delta Y_i$ and $\Delta Y_B$ the increments $Y_i(t)-Y_i(s)$ $\forall i\in N$ and $Y_B(t)-Y_B(s)$, respectively, where $\Delta Y_i \sim SN(0,t-s,\beta_i)$, $\Delta Y_B \sim SN(0,t-s,\beta_B)$, we can write
%
\begin{eqnarray}
R_B &=& \mu_B (t-s) +\sigma_B \Delta Y_B \sim SN(\mu_B  (t-s) ,\sigma_B^2  (t-s) ,\beta_B) \nonumber \\
R_i  &=& \mu_i  (t-s) +\sigma_i \Delta Y_i \sim SN(\mu_i (t-s) ,\sigma_i^2 (t-s),\beta_i) \nonumber
\end{eqnarray}
where $\beta_i$ is as in \eqref{beta1}.

We now indicate by $w=(w_1, \ldots, w_n)$ the vector of portfolio weights, that are the decision variables of the problems addressed in this paper, for which the full investment and no shortselling constraints hold ($\sum_{i=1}^n w_i=1$ and $w_i \geq 0$, respectively).
Furthermore, denoting by $R_P(w)$ the random portfolio return and considering a common assumption in finance that the portfolio return can be expressed as a linear weighted sum of individual stock returns \citep[see, e.g.,][]{canakgoz2009mixed}, we have
\begin{equation*}\label{trak_por}
R_P(w)=\sum_{i=1}^n w_i R_i \, .
\end{equation*}
Since we are interested in examining the difference between the portfolio return and the benchmark return, we can write
\begin{equation}\label{R}
R_B-R_P(w)=  m(w)(t-s)+ \Delta Y(w) \, ,
\end{equation}
where
\begin{equation*}
m(w)=\mu_B - \sum_{i=1}^n w_i \mu_i
\end{equation*}
and
\begingroup\makeatletter\def\f@size{12}\check@mathfonts
\begin{eqnarray}
 \Delta Y(w)  &=& \sigma_B \Delta Y_B -\sum_{i=1}^n w_i \sigma_i \Delta Y_i \nonumber \\
   &=& \biggl(\sigma_B-\sum_{i=1}^n w_i \sigma_i \rho_i \biggr) \Delta Y_B - \sum_{i=1}^n w_i \sigma_i \sqrt{\biggl( 1-\frac{2 \delta_B^2}{\pi} \biggr)(1-\rho_i^2)} \, \Delta W_i \label{Y_new2}
\end{eqnarray}
\endgroup
Then, for any $0 \leq s < t$ we define the \emph{ex-post} tracking error as
\begin{equation}\label{TE_p}
TE^{\,(post)}_{s}(R_B^{s \rightarrow t}-R_P^{s \rightarrow t}(w))=\sqrt{\E[ \left( R_B^{s \rightarrow t}-R_P^{s \rightarrow t}(w) \right)^2]},
\end{equation}
while, for an out-of-sample analysis, we introduce the \emph{ex-ante} tracking error computed in a future time $t > s$ as
 \begin{equation}\label{TE_a}
TE^{\,(ante)}_{s}(R_B^{s \rightarrow t}-R_P^{s \rightarrow t}(w))=\sqrt{\E[ \left( R_B^{s \rightarrow t}-R_P^{s \rightarrow t}(w) \right)^2 \mid \F_s]} \, ,
\end{equation}
where $\F_{s}$ is the filtration at a fixed time $s$.

\noindent
As shown in the following, we will select the tracking portfolio focusing on the \emph{ex-post} tracking error as in \eqref{TE_p}, and we will evaluate its out-of-sample performance by means of the \emph{ex-ante} tracking error as in \eqref{TE_a}.

\subsubsection{Features of the tracking portfolio}\label{Sub_results}

Assuming that the benchmark price $S_B$ and the asset prices $S_i$, with $i\in N= \{1, \ldots, n\}$,
are described by Eq. \eqref{I_Si}, in the following theorem we give explicit expressions for
the \emph{ex-post} tracking error \eqref{TE_p}, for the \emph{ex-ante} tracking error \eqref{TE_a},
and for the forecasted tracking portfolio defined as
$R^{F}_{s}(w) = \E[R_{P}(w,t) \mid \F_s]$, where $0 \leq s < t$.
Furthermore, since in Section \ref{Sec:NumSim} we will provide empirical analysis on a real-world weekly data, and we will compute one week ahead \emph{ex-ante} tracking error and \emph{ex-ante} replicating portfolio returns. Thus, we set $s=t-1$.

\begin{theorem}\label{prop31}
Consider $0 \leq t-1 < t$ and the difference between the benchmark index and the portfolio
returns as in \eqref{R}.
Then, the following results hold.
\begin{itemize}
\item[i)] The \emph{ex-post} tracking error \eqref{TE_p} is given by
\begingroup\makeatletter\def\f@size{12}\check@mathfonts
\begin{eqnarray}
    TE^{\,(post)}_{t-1}(w) &=& \biggl( m^2(w) + 2 m(w) \biggl(\sigma_B-\sum_{i=1}^n w_i \sigma_i \rho_i \biggr) \delta_B \sqrt{\frac{2}{\pi}}  \nonumber \\
   &+&   \biggl(\sigma_B-\sum_{i=1}^n w_i \sigma_i \rho_i \biggr)^2  + \sum_{i=1}^n w_i^2 \sigma_i^2 \biggl( 1-\frac{2 \delta_B^2}{\pi} \biggr)(1-\rho_i^2) \biggr)^{\frac{1}{2}} \label{Eq:ExPostTE}
\end{eqnarray}
\endgroup
\item[ii)] The \emph{ex-ante} tracking error \eqref{TE_a} reads
\begingroup\makeatletter\def\f@size{11}\check@mathfonts
\begin{eqnarray}
 TE^{\,(ante)}_{t-1}(w)  &=& \biggl( m^2(w)+2 m(w)  \biggl(\sigma_B-\sum_{i=1}^n w_i \sigma_i \rho_i \biggr) \biggl[  2u_2^B(t-1)\biggl(\Phi\biggl( \frac{u_2^B(t-1)}{\delta_B }\biggr)-1\biggr) \nonumber \\
  &+&  \delta_B \sqrt{\frac{2 }{\pi}}  \phi\biggl( \frac{u_2^B(t-1)}{\delta_B}\biggr)\biggr] + \biggl(\sigma_B - \sum_{i=1}^n w_i \sigma_i \rho_i \biggr)^2 \nonumber \\
   &\cdot&  \biggl\{ 1+u_2^2(t-1)-\biggl( u_2(t-1)\biggl(2\Phi\biggl( \frac{u_2(t-1)}{\delta_B}\biggr)-1\biggr) + \delta_B \sqrt{\frac{2}{\pi}}\cdot\phi\biggl( \frac{u_2(t-1)}{\delta_B}\biggr)\biggr)^2\biggr\}
   \nonumber \\
     &+& \sum_{i=1}^n w^2_i\sigma^2_i \biggl( 1-\frac{2 \delta^{2}_B}{\pi}\biggr)(1-\rho^2_i) + \biggl\{\biggl(\sigma_B -\sum_{i=1}^n w_i \sigma_i \rho_i \biggr)    \nonumber \\
     &\cdot& \biggl[  2u_2^B(t-1)\biggl(\Phi\biggl( \frac{u_2^B(t-1)}{\delta_B }\biggr)-1\biggr) + \delta_B \sqrt{\frac{2 }{\pi}} \cdot \phi\biggl( \frac{u_2^B(t-1)}{\delta_B}\biggr)\biggr]  \biggr\}^2 \biggr)^{\frac{1}{2}} \label{VarRts}
\end{eqnarray}
\endgroup
where $u_{2}^{B}(t-1)$ is a realization of
$\delta_B  \lvert W_{2}^{B}(t-1) \rvert$.
\item[iii)] The forecasted tracking portfolio return $R^{F}(w)$ is given by
\begingroup\makeatletter\def\f@size{11}\check@mathfonts
\begin{eqnarray}
  R^{F}(w) &=& \E[R_{P}(w) \mid \F_{t-1}] \nonumber \\
    &=& \sum_{i=1}^n w_i \mu_i + \biggl(\sum_{i=1}^n w_i \sigma_i \rho_i \biggr) \biggl[2u_2^B(t-1)\biggl(\Phi\biggl( \frac{u_2^B(t-1)}{\delta }\biggr)-1\biggr)+ \delta_B \sqrt{\frac{2 }{\pi}}  \phi\biggl( \frac{u_2^B(t-1)}{\delta_B}\biggr) \biggr] ,\nonumber \\
   \label{ERts}
\end{eqnarray}
\endgroup
where $u_{2}^{B}(t-1)$ is a realization of
$\delta_B  \lvert W_{2}^{B}(t-1) \rvert$.

%
\end{itemize}
\end{theorem}
\begin{proof}
\hfill \break
\begin{itemize}
\item[i)] We first observe that
\begingroup\makeatletter\def\f@size{11}\check@mathfonts
\begin{eqnarray}
  \E[ \left( R_B^{t-1\to t}-R^{t-1\to t}_P(w) \right)^2] &=& \E[ \left( m(w)+ \Delta Y(w) \right)^2] \nonumber \\
   &=&  m^2(w) + 2 m(w) \E[ \Delta Y(w)] + \Var[\Delta Y(w)] + (\E[ \Delta Y(w)])^2 \nonumber \\
   \label{proof:tr}
\end{eqnarray}
\endgroup
where $\Delta Y(w)$ is as in \eqref{Y_new2}. 

\noindent
Since $\Delta Y_B$ and $\Delta W_i$ are independent random variables, we have
\begingroup\makeatletter\def\f@size{11}\check@mathfonts
\begin{eqnarray}
  \Var[\Delta Y(w)] &=&  \biggl(\sigma_B-\sum_{i=1}^n w_i \sigma_i \rho_i \biggr)^2 \Var[ \Delta Y_B] + \sum_{i=1}^n w_i^2 \sigma_i^2 \biggl( 1-\frac{2 \delta_B^2}{\pi} \biggr)(1-\rho_i^2) \, \Var[\Delta W_i] \nonumber \\
   &=&  \biggl(\sigma_B-\sum_{i=1}^n w_i \sigma_i \rho_i \biggr)^2 \biggl( 1-\frac{2 \delta_B^2}{\pi} \biggr) + \sum_{i=1}^n w_i^2 \sigma_i^2 \biggl( 1-\frac{2 \delta_B^2}{\pi} \biggr)(1-\rho_i^2) \, , \label{proof:Var_DY}
\end{eqnarray}
\endgroup
where from Definition \ref{SBM} and Lemma  \ref{prop:moment-generating} we know that
$\Var[\Delta Y_B]=\bigl( 1-\frac{2 \delta_B^2}{\pi} \bigr)$.
Furthermore,
\begin{equation} \label{proof:E_DY}
    \E[\Delta Y(w)] = \biggl(\sigma_B-\sum_{i=1}^n w_i \sigma_i \rho_i \biggr) \delta_B \sqrt{\frac{2}{\pi}}.
\end{equation}
Therefore, substituting \eqref{proof:Var_DY} and \eqref{proof:E_DY} in \eqref{proof:tr}, we obtain Expression \eqref{Eq:ExPostTE}.
\item[ii)] Using Eq. \eqref{R},
we can write
\begingroup\makeatletter\def\f@size{11}\check@mathfonts
\begin{equation} \label{proof:TE_ante}
\E\biggl[ \left( R_B^{t-1\to t}-R^{t-1\to t}_P(w) \right)^2 \mid \F_{t-1}\biggr]= m^2(w)+2 m(w) \E[\Delta Y(w) \mid \F_{t-1}]+\E[\Delta Y^2(w) \mid \F_{t-1}].
\end{equation}
\endgroup
Furthermore,  from \eqref{Y_new2} we have
\begingroup\makeatletter\def\f@size{12}\check@mathfonts
\begin{eqnarray}
 \E[\Delta Y(w) \mid \F_{t-1}]  &=& \biggl( \sigma_B - \sum_{i=1}^n w_i \sigma_i \rho_i \biggr)\E[\Delta Y_B \mid \F_{t-1}] \nonumber \\
   &-&  \sum_{i=1}^n w_i \sigma_i \sqrt{\biggl(1-\frac{2\delta^2_B}{\pi} \biggr)(1-\rho^2_i)} \, \E[\Delta W_i \mid \F_{t-1}] \, , \nonumber
\end{eqnarray}
\endgroup
and
\begingroup\makeatletter\def\f@size{12}\check@mathfonts
$$
\E[\Delta Y^2(w) \mid \F_{t-1}]=\Var\![\Delta Y(w) \mid \F_{t-1}]+\bigl( \E[\Delta Y(w) \mid \F_{t-1}]\bigr)^2,
$$
\endgroup
with
\begingroup\makeatletter\def\f@size{11}\check@mathfonts
\begin{eqnarray}
 \Var[\Delta Y(w) \mid \F_{t-1}]  &=& \biggl(\sigma_B - \sum_{i=1}^n  w_i \sigma_i \rho_i \biggr)^{\!\!2} \Var[\Delta Y_B \mid \F_{t-1}] \nonumber \\
   &+& \sum_{i=1}^n w^2_i \sigma^2_i \biggl( 1-\frac{2\delta^2_B}{\pi}\biggr)(1-\rho^2_i) \Var\![\Delta W_i \mid \F_{t-1}] \, . \nonumber
\end{eqnarray}
%
\endgroup
%
%
From the properties of the Brownian motion, it is clear that
$$
\E[\Delta W_i \mid \F_{t-1}]=\E[\Delta W_i]=0, \quad
\Var[\Delta W_i \mid \F_{t-1}]=\Var[\Delta W_i]=1.
$$
Exploiting Theorem \ref{prop22}, we can evaluate $\E[\Delta Y_B \mid \F_{t-1}]$ and $\Var\![\Delta Y_B\!\mid\!\F_{t-1}]$ as follows
\begingroup\makeatletter\def\f@size{11}\check@mathfonts
\begin{eqnarray}
  \E[\Delta Y_B \mid \F_{t-1}] &=& \E[Y_B \mid \F_{t-1}]-(u^B_1(t-1)+u^B_2(t-1)) \nonumber \\
   &=& u_1^B(t-1)-u_2^B(t-1)+2 u_2^B(t-1)\Phi\biggl( \frac{u_2^B(t-1)}{\delta_B }\biggr) \nonumber \\
   &+&  \delta_B \sqrt{\frac{2 }{\pi}} \cdot \phi\biggl( \frac{u_2^B(t-1)}{\delta_B}\biggr)-(u^B_1(t-1)+u^B_2(t-1)) \nonumber \\
   &=& 2u_2^B(t-1)\biggl(\Phi\biggl( \frac{u_2^B(t-1)}{\delta_B }\biggr)-1\biggr)+ \delta_B \sqrt{\frac{2 }{\pi}} \cdot \phi\biggl( \frac{u_2^B(t-1)}{\delta_B}\biggr) \, ; \nonumber \\
    \label{EDYB}
\end{eqnarray}
\endgroup
\begingroup\makeatletter\def\f@size{11}\check@mathfonts
\begin{eqnarray}
 \Var[\Delta Y_B \mid \F_{t-1}]  &=& \Var[Y_B(t) \mid \F_{t-1}] + \Var[Y_B(t-1) \mid \F_{t-1}] \nonumber \\
   &=& \Var[Y_B(t) \mid \F_{t-1}] \nonumber \\
   &=& 1+u_2^2(t-1)-\biggl( u_2(t-1)\biggl(2\Phi\biggl( \frac{u_2(t-1)}{\delta_B}\biggr)-1\biggr) \nonumber \\
     &+& \delta_B\sqrt{\frac{2}{\pi}} \phi\biggl( \frac{u_2(t-1)}{\delta_B}\biggr)\biggr)^2 \, , \label{proof:VarF}
\end{eqnarray}
\endgroup
where $u_{1}^{B}(t-1)$ and $u_{2}^{B}(t-1)$
are realizations of
$\sqrt{1-\delta^2_B} \, W_{1}^{B}(t-1)$ and $\delta_B  \lvert W_{2}^{B}(t-1) \rvert$, respectively.
Therefore, substituting \eqref{EDYB} and \eqref{proof:VarF} in \eqref{proof:TE_ante}, we obtain Expression \eqref{VarRts}.

\item[iii)] The forecasted tracking portfolio return $R^F(w)$ 
can be obtained as follows
\begingroup\makeatletter\def\f@size{11}\check@mathfonts
\begin{eqnarray}
 R^F(w) &=& \E[R^{t-1 \to t}_P(w) \mid \F_{t-1}] \nonumber \\
 &=& \sum_{i=1}^n w_i \mu_i + \biggl(\sum_{i=1}^n  w_i \sigma_i \rho_i \biggr) \E[\Delta Y_B \mid \F_{t-1}] \nonumber \\
   &+& \sum_{i=1}^n  w_i \sigma_i \sqrt{\biggl( 1-\frac{2 \delta^2_B}{\pi}\biggr)(1-\rho_i)} \, \E[\Delta W_i \mid \F_{t-1}], \nonumber
\end{eqnarray}
\endgroup
where $\E[\Delta W_i \mid \F_{t-1}]=0$ and $\E[\Delta Y_B \mid \F_{t-1}]$ is as in \eqref{EDYB}. 
\end{itemize}
\end{proof}

\begin{corollary}\label{hPCA_normal}
If $\beta_B$ is equal to zero, our framework reduces to the particular case of normal distributions.
Indeed, from \eqref{delta} one has $\delta_B=0$, and consequently,
$$
Y_B= W_1^{B}, \qquad Y_i = \rho_i Y_B + \sqrt{(1-\rho^2_i)} \, W_i \qquad \forall i \in\{1,...,n\} \, .
$$
Observe that, in this case, $Y_B$ and $Y_i$ are Brownian motions with correlation $Corr(Y_B,Y_i)=\rho_i$.
Then, for any $0\leq t-1 < t$, the \emph{ex-post} tracking error, the \emph{ex-ante} tracking error, and the forecasted index returns (provided by Theorem \ref{prop31}) reduce to
\begingroup\makeatletter\def\f@size{12}\check@mathfonts
\begin{equation} \label{Eq:ExPostTE_Normal}
    TE^{\,(post)}_{t-1}(w)= \biggl( m^2(w) + \biggl(\sigma_B-\sum_{i=1}^n w_i \sigma_i \rho_i \biggr)^2  + \sum_{i=1}^n w_i^2 \sigma_i^2 (1-\rho_i^2) \biggr)^{\frac{1}{2}}
\end{equation}
\endgroup
\begingroup\makeatletter\def\f@size{12}\check@mathfonts
\begin{eqnarray}
  TE_{t-1}^{\,(ante)} &=& \biggl(m^2(w)+
   \biggl(\sigma_B - \sum_{i=1}^n  w_i\sigma_i \rho_i \biggr)^2
   + \sum_{i=1}^n  w^2_i \sigma^2_i (1-\rho^2_i)\biggr)^{\frac{1}{2}}
    \label{TEante_norm}
\end{eqnarray}
\endgroup
and
\begin{equation}\label{prev_norm}
    R^F(w) = \sum_{i=1}^n w_i \mu_i
\end{equation}
respectively.
\end{corollary}

In the next section, we propose the \emph{hybrid} PCA strategy,
where we consider both the cases of normal and skew distributions.

\section{Portfolio selection models for index tracking}\label{Sec:algo}

In this section, we describe two approaches to portfolio selection aimed at replicating a given benchmark:
the \emph{hybrid} PCA strategy, where we consider both normal and skew distributed markets (see Section \ref{hPCA}), and the basic index tracking strategy, where we consider two variants, a standard one typically used in the literature \citep[see, e.g.,][]{scozzari2013exact}, and one used by some practitioners
 (see Section \ref{Sec:BaseMod}).

\subsection{\textbf{Hybrid PCA strategy}}\label{hPCA}

As mentioned in the introduction, the Index Tracking (IT) strategy consists of selecting a small number of assets that replicate a certain benchmark as closely as possible.
To this end, we propose here a novel procedure to tackle the IT problem, called hybrid Principal Component Analysis (hPCA) that we  apply both to normal and skew distributions.

More precisely, we perform a PCA for each pair of random variables $R_B$ and $R_i$ with $i=1, \ldots n$,
thus obtaining the following decomposition
\begin{eqnarray}
  R_{B} &=& \alpha_B + \gamma_{11}^{i} Z_{1}^{i}  + \gamma_{12}^{i}  Z_{2}^{i}  \label{eq:R_B} \\
  R_{i} &=& \alpha_i + \gamma_{21}^{i} Z_{1}^{i} + \gamma_{22}^{i} Z_{2}^{i} \, . \label{eq:r_i}
\end{eqnarray}
 In the case of normal markets, we have
 \begin{align}\label{eq:ParNormal}
 \alpha_B  &= \E[R_{B}]         &  \gamma_{11}^{i} &=e_{11}^{i} \sqrt{\lambda_{1}^{i} }              &  \gamma_{12}^{i} &=e_{12}^{i}  \sqrt{\lambda_{2}^{i} } \\
 \alpha_i  &= \E[R_{i}]         &  \gamma_{21}^{i} &= e_{21}^{i} \sqrt{\lambda_{1}^{i}}   &  \gamma_{22}^{i} &=e_{22}^{i} \sqrt{\lambda_{2}^{i}} \\
 Z_{1}^{i} &  \sim N(0,1)                &  Z_{2}^{i}      &  \sim N(0,1)          &  
 \end{align}
 where $Z_{1}^{i}$ and $Z_{2}^{i}$ are independent and identically distributed (i.i.d.) standard normal random variables;
 the vectors $e_{1}^{i} = (e_{11}^{i}, e_{21}^{i})^{T}$ and $e_{2}^{i} = (e_{12}^{i}, e_{22}^{i})^{T}$ are the eigenvectors of the covariance matrix $\Sigma_{i}$ (as in \eqref{Sigma1}) obtained by $R_B$ and $R_i$ that identify the directions of the $1^{st}$ and the $2^{nd}$ principal components;
 $\lambda_{1}^{i}$ and $\lambda_{2}^{i}$ (see \eqref{lam11} and \eqref{lam12}, respectively) are the eigenvalues of $\Sigma_{i}$.

 \noindent
 In the case of skewed markets, for the benchmark-asset principal
component factorization \eqref{eq:R_B}-\eqref{eq:r_i} we have
 \begingroup\makeatletter\def\f@size{12}\check@mathfonts
 \begin{align}\label{eq:ParSkew}
 \alpha_B  &= \E[R_{B}] - \sqrt{\frac{2}{\pi}} \bigl( \gamma_{11}^{i} \delta_{1}^{i}  + \gamma_{12}^{i}   \delta_{2}^{i}  \bigr)      &  \gamma_{11}^{i} &=e_{11}^{i} \sqrt{\widetilde{\lambda}_{1}^{i} }              &  \gamma_{12}^{i} &=e_{12}^{i}  \sqrt{\widetilde{\lambda}_{2}^{i} } \\
 \alpha_i  &= \E[R_{i}] - \sqrt{\frac{2}{\pi}} \bigl( \gamma_{21}^{i} \delta_{1}^{i}  + \gamma_{22}^{i}   \delta_{2}^{i}  \bigr)        &  \gamma_{21}^{i} &= e_{21}^{i} \sqrt{\widetilde{\lambda}_{1}^{i}}   &  \gamma_{22}^{i} &=e_{22}^{i} \sqrt{\widetilde{\lambda}_{2}^{i}}  \\
 Z_{1}^{i} &  \sim SN(0,1, \beta_{1}^{i})                &  Z_{2}^{i}      &  \sim SN(0,1, \beta_{2}^{i})          & 
 \end{align}
 \endgroup
 where $Z_{1}^{i}$ and $Z_{2}^{i}$ are i.i.d. standard skew-normal random variables;
 $\delta_{1}^{i}=\frac{\beta_{1}^{i}}{\sqrt{1+(\beta_{1}^{i})^2}}$ and $\delta_{2}^{i}=\frac{\beta_{2}^{i}}{\sqrt{1+(\beta_{2}^{i})^2}}$;
 the vectors $e_{1}^{i} = (e_{11}^{i}, e_{21}^{i})^{T}$ and $e_{2}^{i} = (e_{12}^{i}, e_{22}^{i})^{T}$ are the eigenvectors of the covariance matrix $\widetilde{\Sigma}_i$ (as in \eqref{Sigma2}), that identify the directions of the $1^{st}$ and the $2^{nd}$ principal components and are the same of the normal case;
 $\widetilde{\lambda}_{1}^{i}$ and $\widetilde{\lambda}_{2}^{i}$ (see \eqref{app:eigenV_Skew}) are the eigenvalues of $\widetilde{\Sigma}_i$.
%
%

\noindent
For  completeness, in Appendix \ref{sec:Eigen}, we report the algebraic calculations to obtain the principal component factorization \eqref{eq:R_B} and \eqref{eq:r_i}.

\subsubsection{Tracking error through benchmark-asset principal component factorization }\label{M1}

In this section we provide the expression of the tracking error obtained by the principal component factorization introduced above.
From \eqref{TE_p}, we have $TE^{\,(post)}(R_{B} - R_{P}(w))=\sqrt{\E[(R_{B} - R_{P}(w))^2]}$.
Now, we can write
\begingroup\makeatletter\def\f@size{12}\check@mathfonts
\begin{eqnarray}
  \E[(R_{B} - R_{P}(w))^2] &=& \E[R_{B}^2]   + \E[R_{P}(w)^2] - 2 \E[R_{B} R_{P}(w)] \nonumber \\
    &=& \breve{\sigma}_B^2 + \breve{\mu}_B^2 + \breve{\sigma}_P^2(w) + \breve{\mu}_P^2(w) - 2 \Cov(R_{B},R_{P}(w)) - 2 \breve{\mu}_B \breve{\mu}_P(w) \nonumber \\
   &=& \breve{\sigma}_B^2 + \breve{\sigma}_P^2(w) + (\breve{\mu}_B - \breve{\mu}_P(w)])^2 - 2 \Cov(R_{B},R_{P}(w))  \nonumber
\end{eqnarray}
\endgroup
where $\breve{\mu}_B=\E[R_{B}]$, $\breve{\sigma}_B^2=\Var[R_{B}]$,
$\breve{\mu}_P(w)=\E[R_{P}(w)]$, $\breve{\sigma}_P^2(w)=\Var[R_{P}(w)]$,
and
\begin{equation*}
  \Cov(R_{B},R_{P}(w))=\Cov(R_{B},\sum_{i=1}^{n} R_i w_i)=\sum_{i=1}^{n}  w_i \Cov(R_{B},R_i) \, .
\end{equation*}

\paragraph{Gaussian distributions}
\hfill
\break
\noindent
Using the principal component factorization \eqref{eq:R_B}-\eqref{eq:r_i} with \eqref{eq:ParNormal}
we have
\begingroup\makeatletter\def\f@size{11}\check@mathfonts
  \begin{eqnarray}
 \Cov(R_{B},R_i) &=& \Cov(\alpha_B + e_{11}^{i} \sqrt{\lambda_{1}^{i} } Z_{1}^{i}  + e_{12}^{i}  \sqrt{\lambda_{2}^{i} } Z_{2}^{i}, \, \alpha_i + e_{21}^{i} \sqrt{\lambda_{1}^{i}} Z_{1}^{i} + e_{22}^{i} \sqrt{\lambda_{2}^{i}} Z_{2}^{i}) \nonumber \\
  &=& e_{11}^{i} e_{21}^{i} \lambda_{1}^{i} + e_{12}^{i} e_{22}^{i} \lambda_{2}^{i} \nonumber \, .
\end{eqnarray}
\endgroup
  Therefore,
\begingroup\makeatletter\def\f@size{12}\check@mathfonts
  \begin{eqnarray}
 \Cov(R_{B},R_{P}(w)) &=& \sum_{i=1}^{n}  w_i e_{11}^{i} e_{21}^{i} \lambda_{1}^{i} + \sum_{i=1}^{n}  w_i e_{12}^{i} e_{22}^{i} \lambda_{2}^{i} \nonumber \\
  &=& \sum_{i=1}^{n}  w_i C_{1}^{i} + \sum_{i=1}^{N}  w_i C_{2}^{i} \label{eq:covariance} \, .
\end{eqnarray}
\endgroup
where, as shown in Appendix \ref{app:NormalCase},
\begingroup\makeatletter\def\f@size{12}\check@mathfonts
  \begin{eqnarray}
 C_{1}^{i} &=& e_{11}^{i} e_{21}^{i} \lambda_{1}^{i} =  \frac{\rho_i\sigma_B\sigma_i(\lambda^i_1-\sigma^2_B)\lambda^i_1}{\rho^2_i\sigma^2_B\sigma^2_i+(\lambda^i_1-\sigma^2_B)^2} \nonumber \\
 C_{2}^{i}  &=& e_{12}^{i} e_{22}^{i} \lambda_{2}^{i} = \frac{\rho_i\sigma_B\sigma_i(\lambda^i_2-\sigma^2_B)\lambda^i_2}{\rho^2_i\sigma^2_B\sigma^2_i+(\lambda^i_2-\sigma^2_B)^2} \nonumber \, .
\end{eqnarray}
\endgroup
Hence, using \eqref{eq:covariance}, we can write
\begingroup\makeatletter\def\f@size{12}\check@mathfonts
\begin{equation} \label{eq:TEviapcf}
  TE^{\,(post)}(R_{B} - R_{P}(w))= \biggl(\breve{\sigma}_B^2 + \breve{\sigma}_P^2(w) + (\breve{\mu}_B - \breve{\mu}_P(w)])^2 - 2 \biggl(  \sum_{i=1}^{n}  w_i C_{1}^{i}  +\sum_{i=1}^{n}  w_i C_{2}^{i}  \biggr) \biggr)^{\frac{1}{2}} \, .
\end{equation}
\endgroup
We first notice that assuming $\rho_i >0$ \citep[as typically happens in practice, see, e.g.,][]{Martens2001, Zhang2020}, since $\lambda^i_1 \geq \lambda^i_2 \geq 0$, $\lambda_{1}^{i} \geq \sigma^2_B$, and $\lambda_{2}^{i} \leq \sigma^2_B$, as shown in Appendix \ref{app:NormalCase}, %
$C_{1}^{i} \geq 0$ and $C_{2}^{i} \leq 0$.
So, according to Expression \eqref{eq:TEviapcf} to decrease TE, the idea is to select the assets that maximize $C_{1}^{i}$ and minimize $C_{2}^{i}$.
Second, we observe that if we consider a portfolio consisting of a single asset with $ \rho_i=1$ and $\sigma^2_i=\sigma^2_B$, the tracking error would be equal to 0.
In terms of benchmark-asset principal component factorization, this means that $\lambda_{1}^{i}=2 \sigma^2_B$
and $\lambda_{2}^{i}=0$, namely $C_{1}^{i}=\sigma^2_B$ and $C_{2}^{i} = 0$.
Therefore, we propose an index tracking strategy without any optimization algorithm where we select $K<n$ assets for which their values of $\lambda_{1}^{i}$ and $\lambda_{2}^{i}$ are closer
to the "optimal" values $2 \sigma^2_B$ and $0$.
A possible way to do this is to select the assets that have the lowest Euclidean distance
\begin{equation*} \label{eq:EuDist}
    d_i=\lVert \mathbf{\lambda^{i}} -\mathbf{\lambda^{0}} \rVert
\end{equation*}
where $\mathbf{\lambda^{i}}=(\lambda_{1}^{i}, \lambda_{2}^{i})$ and
$\mathbf{\lambda^{0}}=(2 \sigma^2_B, 0)$.

\noindent
The financial intuition is to select the factor loadings that best match the index with the selected assets. Furthermore, since PCA is a dimension reduction technique, the first eigenvalues may capture technical indicators such as directionality and market momentum. Recent examples in the literature can be found in \cite{Liang2020}  extracting common factors in commodity futures, and \cite{Zheng2021} for dimension reduction and forecasting.

\paragraph{Skew-normal distributions}
\hfill
\break
\noindent
In the case of a skew-normal returns, with analogous arguments,
using the principal component factorization \eqref{eq:R_B}-\eqref{eq:r_i} with \eqref{eq:ParSkew},
we find that
\begingroup\makeatletter\def\f@size{12}\check@mathfonts
  \begin{eqnarray}
 \Cov(R_{B},R_i) &=& \Cov(\alpha_B + e_{11}^{i} \sqrt{\widetilde{\lambda}_{1}^{i} } Z_{1}^{i}  + e_{12}^{i}  \sqrt{\widetilde{\lambda}_{2}^{i} } Z_{2}^{i}, \, \alpha_i + e_{21}^{i} \sqrt{\widetilde{\lambda}_{1}^{i}} Z_{1}^{i} + e_{22}^{i} \sqrt{\widetilde{\lambda}_{2}^{i}} Z_{2}^{i}) \nonumber \\
  &=& e_{11}^{i} e_{21}^{i} \xi_1^i\widetilde{\lambda}_{1}^{i} + e_{12}^{i} e_{22}^{i} \xi_2^i\widetilde{\lambda}_{2}^{i} \nonumber \, ,
\end{eqnarray}
\endgroup
%
Furthermore, we have that
\begin{eqnarray}
 \xi_1^i  &=& \Var[Z_{1}^{i}] =\biggl( 1-\frac{2(\delta_1^i)^2}{\pi}\biggr) \nonumber \\
  \xi_2^i &=& \Var[Z_{2}^{i}] =\biggl( 1-\frac{2(\delta_2^i)^2}{\pi}\biggr) \nonumber
\end{eqnarray}
%
where $\delta_q^i=\frac{\beta_q^i}{\sqrt{1+(\beta_q^i)^2}}$ with $q=1, 2$.
Hence,
%
  \begin{eqnarray}
 \Cov(R_{B},R_{P}(w)) &=& \sum_{i=1}^{N}  w_i e_{11}^{i} e_{21}^{i} \xi^i_1\widetilde{\lambda}_1^{i} + \sum_{i=1}^{N}  w_i e_{12}^{i} e_{22}^{i} \xi^i_2\widetilde{\lambda}_2^{i} \nonumber \\
  &=& \sum_{i=1}^{N}  w_i \widetilde{C}_{1}^{i} + \sum_{i=1}^{N}  w_i \widetilde{C}_{2}^{i} \nonumber \, .
\end{eqnarray}
%
where
%
  \begin{eqnarray}
 \widetilde{C}_{1}^{i} &=& e_{11}^{i} e_{21}^{i} \xi^i_1 \widetilde{\lambda}_{1}^{i} =  \frac{\rho_i \sigma_B \sigma_i (\lambda^i_1-\sigma^2_B) \xi^i_1 \widetilde{\lambda}^i_1}{\rho^2_i\sigma^2_B\sigma^2_i+(\lambda^i_1-\sigma^2_B)^2} \nonumber \\
 \widetilde{C}_{2}^{i}  &=& e_{12}^{i} e_{22}^{i}
 \xi^i_2\widetilde{\lambda}_{2}^{i} = \frac{\rho_i \sigma_B \sigma_i (\lambda^i_2-\sigma^2_B) \xi^i_2\widetilde{\lambda}^i_2}{\rho^2_i \sigma^2_B \sigma^2_i + (\lambda^i_2-\sigma^2_B)^2} \nonumber \,,
\end{eqnarray}
%
and $\lambda^i_1$ and $\lambda^i_2$ are as in \eqref{lam11} and \eqref{lam12}, respectively.
Hence, the \emph{ex-post} tracking error can be expressed as follows
 %
\begin{equation*} \label{eq:TEviapcf_Skew}
  TE^{\,(post)}(R_{B} - R_{P}(w))= \biggl(\breve{\sigma}_B^2 + \breve{\sigma}_P^2(w) + (\breve{\mu}_B - \breve{\mu}_P(w))^2 - 2 \biggl(  \sum_{i=1}^{N}  w_i \widetilde{C}_{1}^{i}  +\sum_{i=1}^{N}  w_i \widetilde{C}_{2}^{i}  \biggr) \biggr)^{\frac{1}{2}} \, ,
\end{equation*}
%
where, again, $\breve{\mu}_B=\E[R_{B}]$, $\breve{\sigma}_B^2=\Var[R_{B}]$,
$\breve{\mu}_P(w)=\E[R_{P}(w)]$, and $\breve{\sigma}_P^2(w)=\Var[R_{P}(w)]$.
Also in this case assuming $\rho_i >0$, since $\widetilde{\lambda}^i_1 \geq \widetilde{\lambda}^i_2 \geq 0$, $\lambda_{1}^{i} \geq \sigma^2_B$, and $\lambda_{2}^{i} \leq \sigma^2_B$,
$\widetilde{C}_{1}^{i} \geq 0$ and $\widetilde{C}_{2}^{i} \leq 0$.
Thus, following the same rationale illustrated for the normal case,
we select $K<n$ assets such that their values of $\widetilde{\lambda}_{1}^{i}$ and $\widetilde{\lambda}_{2}^{i}$ are closer
to the "optimal" values $\widetilde{\lambda}_{1}^{0}=2 \sigma^2_B \bigl( 1-\frac{2\delta^2_B}{\pi}\bigr)$ and $\widetilde{\lambda}_{2}^{0}=0$.

Namely, we select the assets with the lowest Euclidean distance
\begin{equation} \label{eq:EuDist}
    \widetilde{d}_i=\lVert \mathbf{\widetilde{\lambda}^{i}} -\mathbf{\widetilde{\lambda}^{0}} \rVert
\end{equation}
where $\mathbf{\widetilde{\lambda}^{i}}=(\widetilde{\lambda}_{1}^{i}, \widetilde{\lambda}_{2}^{i})$ and
$\mathbf{\widetilde{\lambda}^{0}}=\bigl(2 \sigma^2_B \bigl( 1-\frac{2\delta^2_B}{\pi}\bigr), 0\bigr)$ (see Appendix \ref{Eigsn}).

\subsubsection{Description of the hybrid PCA procedure for IT}

Below we present a brief description of the hybrid PCA procedure for Index Tracking (IT) purpose,
where we identify the $K$ assets with the lowest values $\widetilde{d}_i$. We define a portfolio where the weights are
decreasing for increasing values of $\widetilde{d}_i$. Then, we compute the \emph{ex-post} and \emph{ex-ante} tracking errors for such a portfolio,
and the future index tracking portfolio returns by Eq. \eqref{ERts}.
More precisely, in the case of skew-normal distributed returns, the hPCA procedure consists of the following steps.
\begin{enumerate}
\item Let $T$ denote the length of the time series analyzed and let $L$ be the length of the in-sample window.
Set $\tau \in [1,T-L+1]$, consider a fixed size rolling window $\mathcal{I}=\{\tau, \tau+1,.., \tau+L-1\}$,
and denote by $r_i(j)$ (with $i\in \{ 1,...,n \}$) and by $r_{B}(j)$ the historical scenarios at time $j\in \mathcal{I}$ of the returns of the asset $i$ and the benchmark, respectively.
\item At time $s=t-1= \tau+L-1$, calibrate on the time-window $\mathcal{I}$ the parameters $\mu_B,\,\sigma_B,\,\beta_B$ of $R_B$,
$\mu_i,\,\sigma_i,\,\beta_i$ of $R_i$, and the parameter $\rho_i$ that measures the dependence between $R_B$ and $R_i$ with $i\in N$.
For this purpose, we use the Maximum Likelihood Estimation (MLE) method, described in Section \ref{sec:MLE}, thus obtaining the calibrated parameters $(\hat \mu_B,\hat \sigma_B,\hat \beta_B)$ and $(\hat \mu_i,\hat \sigma_i,\hat \beta_i)$ (with $i=1, \ldots, n$) for the benchmark and for the assets, respectively.
\item Fix $K \ll n$, the number of assets in the tracking portfolio (e.g., 10 out of 500 assets available in the investment universe).
As described in Section \ref{M1}, we choose the $K$ assets
for which their eigenvalues, obtained from the benchmark-asset principal component factorization, namely $\widetilde{\lambda}_{1}^{i}$ and $\widetilde{\lambda}_{2}^{i}$, are closest
to the ideal values $2 \sigma^2_B \bigl( 1-\frac{2\delta^2_B}{\pi}\bigr)$ and $0$.
More precisely, we select the $K$ assets $\{i_1,i_2,...,i_K \} \subset N$ for which $\widetilde{d}_{i_1} \leq \widetilde{d}_{i_2}  \leq \ldots  \leq \widetilde{d}_{i_K}$.
\item Compute the weights $\hat{w}$ of the tracking portfolio giving decreasing importance to the selected assets for increasing values of $\widetilde{d}_{i}$. In this experiment we consider
\begin{equation} \label{optW2}
\hat{w}_i = \begin{cases}
         \displaystyle \frac{K - h + 1}{\sum_{h=1}^{K} h}=\frac{2(K - h + 1)}{K (K+1)} \quad &\mbox{if} \quad i \in \{i_h\}_{h=1,\ldots, K}  \\
          0 \quad &\mbox{if} \quad i \notin \{i_h\}_{h=1,\ldots, K} \\
     \end{cases}
  \end{equation}
\item Finally, by means of the tracking portfolio $\hat{w}$, compute
the \emph{ex-post} tracking error \eqref{Eq:ExPostTE},
the \emph{ex-ante} tracking error \eqref{VarRts}, and
the forecasted tracking portfolio return \eqref{ERts} provided in Theorem \ref{prop31}.
\end{enumerate}
\medskip

\noindent
In Table \ref{Tab:Algo2} we summarize the  hybrid PCA procedure (pseudocode) for Index Tracking.
\begin{table}[htbp]
	\centering
	\scalebox{0.95}{
	\begin{threeparttable}
		\begin{tabular}{|l|}
			\hline\noalign{\smallskip}
            \green{1.} Fix $T,\,L$ and set $\tau = 1$; \\
            \green{2.} \textbf{while} $\tau < T-L+1$ \\			
			\green{3.} take the observations $r_B(j)$ and $r_i(j)$ (with $i\in N$) for all $j\in \mathcal{I}=\{ \tau, \tau+1,.., \tau+L-1\}$;\\
			\green{4.} calibrate the parameters $(\hat \mu_B,\hat \sigma_B, \hat \beta_B)$ by solving Problem \eqref{cal1}; \\
			\green{5.} compute $\hat \rho_i$ (the Spearman correlation between $r_B$ and $r_i$) and $\hat \beta_i$ by Eq. \eqref{eq:CalBetaBhat};                                       \\
			\green{6.} calibrate $(\hat \mu_i,\hat \sigma_i)$ by solving Problem \eqref{cal2} for all $i\in N$; \\
			\green{7.} find the $K$ assets with the lowest $\widetilde{d}_i$ as in \eqref{eq:EuDist}; \\
			\green{8.} compute  the weights $\hat{w}$ of the tracking portfolio
			by \eqref{optW2}; \\
			\green{9.} compute the \emph{ex-post} and \emph{ex-ante} tracking error by \eqref{Eq:ExPostTE} and \eqref{VarRts}, respectively, and \\
             compute the forecasted tracking portfolio return \eqref{ERts}; \\
			\green{10.} update $t=t+1$; \\
			\green{11.} \textbf{end} \\
			\noalign{\smallskip}\hline
		\end{tabular}
	\end{threeparttable}
	}
		\caption{Pseudocode of the hybrid PCA procedure}
	\label{Tab:Algo2}
\end{table}

\noindent
In the case of normally distributed markets, we follow a procedure similar to that described in Table \ref{Tab:Algo2}.
More precisely,
we estimate the parameters $(\hat \mu_B,\hat \sigma_B)$ and $(\hat \mu_i,\hat \sigma_i)$ $(i\in N)$ through the sample mean and the sample standard deviation of $r_B(j)$ and $r_i(j)$, respectively, for any $j\in \mathcal{I}$.
Moreover, the \emph{ex-post} tracking error, the \emph{ex-ante} tracking error and the forecasted tracking portfolio returns are given by \eqref{Eq:ExPostTE_Normal}, \eqref{TEante_norm} and \eqref{prev_norm}, respectively, in Corollary \ref{hPCA_normal}.

\subsubsection{Calibration of the parameters through MLE}\label{sec:MLE}

In this section we show, under the assumption of skew-normal distributions of returns, how to calibrate the model's parameters $(\mu_B, \sigma_B, \beta_B)$ and $(\mu_i, \sigma_i, \beta_i)$ for all $i \in N$ using the Maximum Likelihood Estimation (MLE) method.

\noindent
Let $r_B(j)$ denote the observations of the benchmark index return $R_B \sim SN(\mu_B, \sigma_B^2, \beta_B)$, for any $j \in \mathcal{I}=\{ \tau, \tau+1,.., \tau+L-1\}$.
Following the results provided by Azzalini in \cite{Azzalini2021},
we can write the likelihood function of $R_B$ as
$$
L_B(\mu_B,\sigma_B,\beta_B)=
\frac{2^L}{\sigma_B^L} \, \prod_{j=\tau}^{\tau+L-1} \phi\biggl( \frac{r_B(j) - \mu_B}{\sigma_B} \biggr) \Phi\biggl(\beta_B\, \frac{r_B(j) -\mu_B}{\sigma_B} \biggr),
$$
and, therefore the estimated parameters $(\hat \mu_B,\hat \sigma_B, \hat \beta_B)$ can be
found by solving the following optimization problem
\begin{equation}\label{cal1}
(\hat \mu_B, \hat \sigma_B, \hat \beta_B) = \displaystyle \argmax_{\mu_B, \sigma_B,\beta_B} \ln L_B(\mu_B,\sigma_B,\beta_B)  \, .
\end{equation}
Then, once estimated $\hat \beta_B$, we can compute $\hat \beta_i$ from Eq. \eqref{beta1}, namely
\begin{equation}\label{eq:CalBetaBhat}
\widehat{\beta}_i=\frac{\hat \beta_B}{\sqrt{1+(1+\beta_B^2)\bigl( 1-\frac{2\hat \delta_B^2}{\pi}\bigr)\bigl(\frac{1}{\hat \rho_i^2}-1\bigr)}} \, ,
\end{equation}
where $\hat \delta_B = \frac{\hat \beta_B}{\sqrt{1+\hat \beta_B^2}}$,
and $\hat \rho_i$ is the Spearman correlation between $R_B$ and $R_i$.

\noindent
From Proposition \ref{prop21}, we can obtain the likelihood function of $R_i$ conditioned to $\hat \beta_i$
%
\begin{eqnarray}
 L_i(\mu_i,\sigma_i)  &=& \frac{2^L}{\biggl(\sigma_i \sqrt{\hat \rho^2_i+\bigl( 1-\frac{2\hat \delta_B^2}{\pi}\bigr)(1-\hat \rho_i^2\bigr)} \biggr)^L } \prod_{j= \tau}^{\tau+L-1} \phi\biggl( \frac{r_i(j)-\mu_i}{\sigma_i \sqrt{\hat \rho^2_i+\bigl( 1-\frac{2\hat \delta_B^2}{\pi}\bigr)(1-\hat \rho_i^2\bigr)}} \biggr) \nonumber \\
   &\cdot&  \Phi\biggl(\hat \beta_i\, \frac{r_i(j)-\mu_i}{\sigma_i \sqrt{\hat \rho^2_i+\bigl( 1-\frac{2\hat \delta_B^2}{\pi}\bigr)(1-\hat \rho_i^2\bigr)}} \biggr), \nonumber
\end{eqnarray}
%
by which we can find $\hat \mu_i$ and $\hat \sigma_i$ as follows
\begin{equation}\label{cal2}
(\hat \mu_i,\hat \sigma_i)=\argmax_{\mu_i, \sigma_i} \ln L_i(\mu_i,\sigma_i) \, .
\end{equation}

\subsection{The portfolio optimization models for index tracking}\label{Sec:BaseMod}

In this section, we provide the mathematical formulation of the standard Index Tracking (IT) optimization problem based on the minimization of the tracking error in terms of objective function (Section \ref{Sec:StandardIT}) that we call \emph{baseline} index tracking approach and we compare it with our approach. Furthermore, we also present the IT strategy used in the financial industry that we call \emph{practitioner} index tracking approach (Section \ref{Sec:practitioners}).

\subsubsection{The baseline index tracking approach}\label{Sec:StandardIT}

For each rolling time-window $\mathcal{I}=\{\tau, \tau+1,.., \tau +L-1\}$ with $\tau \in [1,T-L+1]$,
we select the optimal tracking portfolio obtained by solving the following Mixed Integer Quadratic
Programming problem
\begin{equation}\label{eq:baseline}
\left\lbrace \begin{array}{lll}
\displaystyle\min_{w} & TE^{\,(post)}(w)=\sqrt{\displaystyle \frac{1}{L} \sum_{j \in \mathcal{I}}\bigl(r^{B}_{j} - r^{P}_{j}(w) \bigr)^2} & \\
\mbox{s.t.} & & \\
& \displaystyle \sum_{i=1}^n w_i=1  &  \\
&  \sum_{i=1}^N y_i = K  & \\
& 0 \leq w_i \leq y_i  & i=1, \ldots, n \\
& y_i \in \{0, 1\}  & i=1, \ldots, n
\end{array} \right.
\end{equation}
where

\noindent
$r^{B}_{j}$ represents the historical scenario of the benchmark index return at time $j \in \mathcal{I}$;

\noindent
$r_{i, j}$ is the historical return of asset $i$ at time $j$;

\noindent
$w$ is the vector of the portfolio weights whose elements $w_i$ are the fractions of a given capital invested in asset $i$;

\noindent
$r^{P}_{j}(w) = \sum_{i=1}^n w_i r_{i, j}$ represents the historical scenario of the portfolio return at time $j \in \mathcal{I}$;

\noindent
$n$ is the number of assets available in the investment universe;

\noindent
$K$ is the fixed number of assets selected in the tracking portfolio (in our empirical analysis $K=10$).

\subsubsection{A  practitioner approach to index tracking}\label{Sec:practitioners}

Here we briefly describe a common index tracking optimization model adopted in the financial industry that we call \emph{practitioner} IT approach and use it as a second baseline model for comparison purposes.

\noindent
As mentioned above, for building a passive  portfolio, such as an indexed one,
an asset manager seeks to manage his exposure to the benchmark by selecting the least number of securities. So, first, the benchmark is decomposed into sectors and, once the weight of each sector is known, the trader chooses the stocks that he believes will perform best in each sector.
The first step is called asset allocation, the second is called stock picking. Here we focus on asset allocation, that is, on determining the weights of the sectors.

The above \emph{practitioner} index tracking approach consists in considering as  constituents the sub-indices (sectors) of the benchmark. Typically, the number of constituents of this new investment universe is chosen exactly equal to $K$.

\noindent
Then, for the \emph{practitioner} index tracking approach we compute the optimal weights of the portfolio replicating the benchmark by solving the following convex quadratic programming problem
\begin{equation*}\label{Opt_probl2}
\left\lbrace \begin{array}{lll}
\displaystyle\min_{w} & TE^{\,(post)}(w)=\sqrt{\displaystyle \frac{1}{L} \sum_{j \in \mathcal{I}}\bigl(r^{B}_{j} - \widetilde{r}^{P}_{j}(w) \bigr)^2} & \\
\mbox{s.t.} & & \\
& \displaystyle \sum_{i=1}^K w_i=1  &  \\
& w_i \geq 0  & i=1, \ldots, K
\end{array} \right.
\end{equation*}
where

\noindent
for a specific $\tau \in [1,T-L+1]$, $\mathcal{I}=\{\tau, \tau+1,.., \tau +L-1\}$ represents the rolling time window;

\noindent
$r^{B}_{j}$ represents the historical scenario of the benchmark index return at time $j \in \mathcal{I}$;

\noindent
$\widetilde{r}_{i, j}$ is the historical return of sector $i$ at time $j$;

\noindent
$w$ is the vector of the portfolio weights whose elements $w_i$ are the fractions of a given capital invested in sector $i$;

\noindent
$\widetilde{r}^{P}_{j}(w) = \sum_{i=1}^n w_i \widetilde{r}_{i, j}$ represents the historical scenario of the portfolio return of sub-indices at time $j \in \mathcal{I}$;

\noindent
$K$ is the number of sub-indices available in the market.

\section{Empirical analysis} \label{Sec:NumSim}

Here we provide an empirical analysis that compares the IT approaches described in Section \ref{Sec:algo}, both in terms of computational efficiency and performance.
The experiments have been conducted on the S\&P 500 dataset, which consists of weekly prices retrieved from 7 January 2005 to 29 May 2020, for a total of $T=804$ observations.
For the sake of space and readability, details about the analyzed dataset are available in Appendix \ref{app:DataInfo} where we report the list of the S\&P 500 constituents as well as the $K=10$ sectors (sub-indices). 

\noindent
For the out-of-sample performance analysis, we adopt a rolling time-window, and we allow for the possibility of rebalancing the portfolio composition during the holding period at fixed intervals. In this study, we set 1 year ($L=52$) for the in-sample window and 1 week both for the rebalancing interval and the holding period. On these portfolios we compute the \emph{ex-post} tracking error \eqref{Eq:ExPostTE},
the \emph{ex-ante} tracking error \eqref{VarRts}, and
the forecasted tracking portfolio returns \eqref{ERts}.
The cardinality $K$ of the analyzed tracking portfolios is set equal to 10.

All the procedures have been implemented in PYTHON 3.10 and have been executed on a laptop with an Intel(R) Core(TM)
i7-4800HQ CPU @ 2.6 GHz processor and 16,00 GB of RAM.
Furthermore, the Mixed Integer Quadratic Programming problem \eqref{eq:baseline}
is solved by using GUROBI 9.5 called from PYTHON \citep{Gurobi95}.

\subsection{\textbf{Computational results}}


In this section, we report and compare the performance analysis obtained by the hybrid PCA (hPCA) strategies (with normal and skew-normal assumptions, see Section \ref{hPCA}) and by the Index Tracking approaches (described in Section \ref{Sec:BaseMod}).
\noindent
Figure \ref{FigU1} depicts the \emph{ex-post} tracking error computed by the normal and skew-normal hPCA strategies, and by the baseline and practitioners’ IT approaches (see Sections \ref{Sec:StandardIT} ad \ref{Sec:practitioners}, respectively).
From that comparison the \emph{ex-post} tracking error of the hPCA, under the assumption of skew-normal returns, is globally lower than those of the competitor models.
\begin{figure}[!htbp]
      \centering
         \includegraphics[width=1\textwidth]{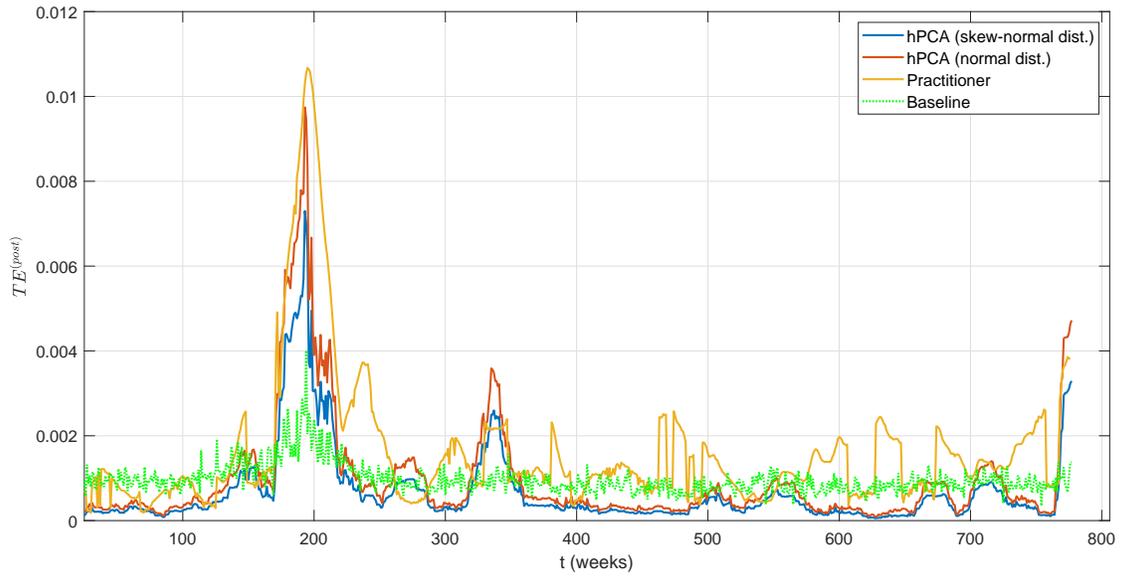}
         \caption{\emph{Ex-post} Tracking Error (TE) for the skew-normal hPCA (blue), the normal hPCA (red), the baseline (green), the practitioner (yellow) strategies.}
        \label{FigU1}
\end{figure}

\noindent
Similarly, Figure \ref{FigU2} shows the \emph{ex-ante} tracking error for all considered models.
In addition, to reveal further information on the performance of the different approaches, the \emph{ex-ante} tracking error provides an insight into portfolio construction and risk budgeting. This is because asset managers use tracking error as a measure of the risk of deviating from the benchmark. Since an index tracking portfolio is assigned a risk/reward target, a sudden change in tracking error requires a swiftly rebalance of the constituents.
From this point of view, the frenzy changes of the baseline model \eqref{eq:baseline} causes  disruptions in portfolio management that are not evident in the classical analysis of turnover often used in the literature.
\begin{figure}[!htbp]
      \centering
         \includegraphics[width=1\textwidth]{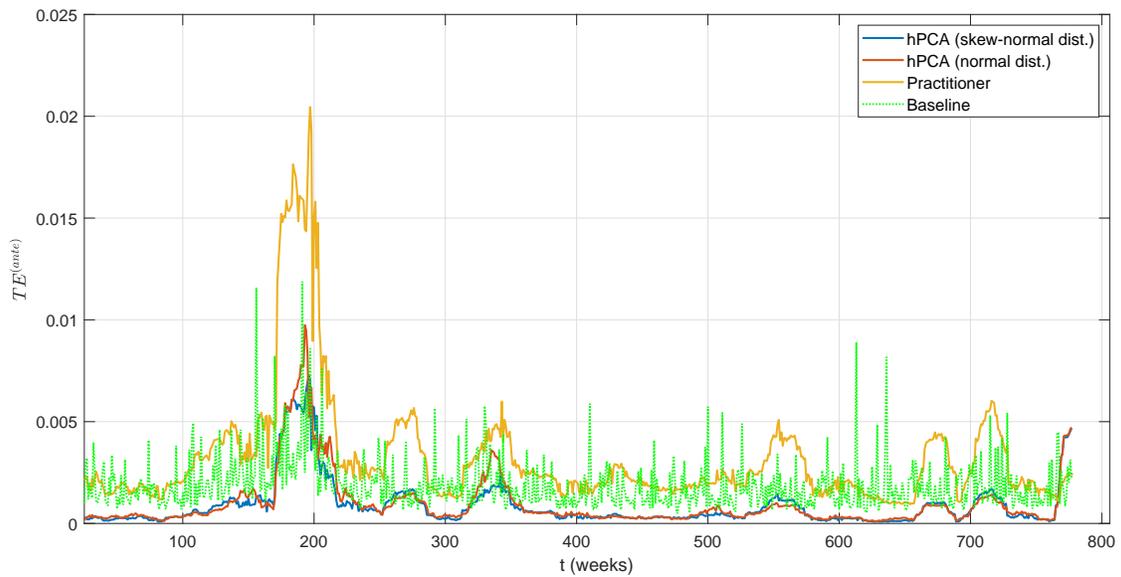}
         \caption{\emph{Ex-ante} Tracking Error (TE) for the skew-normal hPCA (blue), the normal hPCA (red), the baseline (green), the practitioner (yellow) strategies.}

        \label{FigU2}
\end{figure}

\noindent
Globally, comparing both \emph{ex-post} and \emph{ex-ante} tracking error between models,
Table  \ref{Tab_RelTE} reports the number of times (in percentage) that the tracking error of the portfolio constructed with the skew-normal hPCA strategy is lower than that obtained from the other models. We observe that the skew-normal hPCA strategy always shows values greater than 90\% thus confirming the validity of the suggested approach.
\begin{table}[htbp!]
\centering
  \begin{adjustbox}{width=1\textwidth}
\begin{threeparttable}

    \begin{tabular}{|c|c|c|c|}
\toprule
\textbf{TE}  &

\begin{tabular}[c]{@{}c@{}}\textbf{skew-normal hPCA} \\ \textbf{vs. normal hPCA}
\end{tabular}
&
\begin{tabular}[c]{@{}c@{}}\textbf{skew-normal hPCA} \\
\textbf{vs. \textbf{Practitioner}}
\end{tabular}
&
\begin{tabular}[c]{@{}c@{}}\textbf{skew-normal hPCA} \\
\textbf{vs. Baseline}
\end{tabular} \\
\hline
$TE^{\,(post)}$  &  94.01\% & 93.13\%  & 90.51\%  \\
\hline
$TE^{\,(ante)}$ &    96.91\%    & 92.26\%  & 89.51\%  \\
 \bottomrule
\end{tabular}

\end{threeparttable}
  \end{adjustbox}
  \caption{Number of times (in percentage) that the tracking error of the skew-normal hPCA portfolio is lower than that obtained from the other strategies.}\label{Tab_RelTE}
\end{table}


So far we discussed the tracking error, which, as mentioned in Section \ref{Sec:LitRev}, may be seen as the risk of not investing in the benchmark. The other side of the coin is the reward.
Figure \ref{FigU3}  displays the differences between the  returns of the benchmark and the returns of the replicating portfolios.
\begin{figure}[!htbp]
      \centering
         \includegraphics[width=1\textwidth]{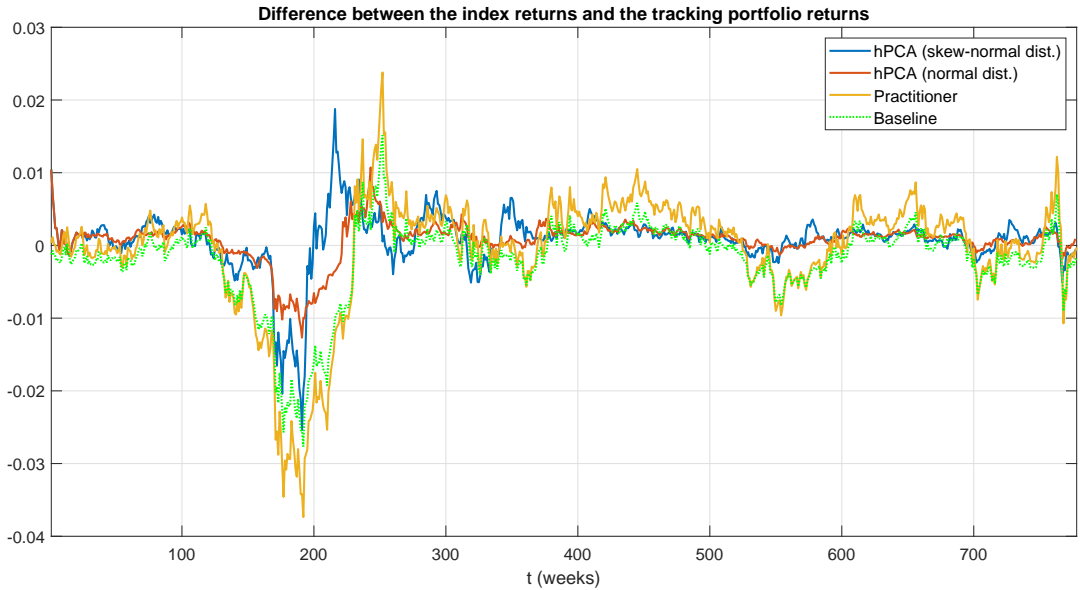}
         \caption{Difference between the benchmark returns and the returns of the tracking portfolios obtained by the skew-normal hPCA (blue), the normal hPCA (red), the baseline (green), the practitioner (yellow) strategies.}
        \label{FigU3}
\end{figure}
Observe that while Figure \ref{FigU1} displays lower \emph{ex-post} tracking error during turbulent periods (e.g., financial crises of 2007-9) for the baseline model, actually,
Figure \ref{FigU3} shows that the baseline model is performing worse than the skew-normal hPCA model.
Still, by looking at the said figure, similar behaviors can be observed in other turbulent occasions.
A detailed account is available in Table \ref{Tab_over_TE}, where, for each non-overlapping subperiods,
we report the average annualized excess return of a portfolio w.r.t. the benchmark.
We can observe that the hPCA approach generally shows
better performance.
\begin{table}[htbp!]
\small
\centering
      \begin{adjustbox}{width=1\textwidth}
\begin{threeparttable}

    \begin{tabular}{|c|c|c|c|c|}
\toprule
\textbf{Interval}  & \textbf{skew-normal hPCA} & \textbf{normal hPCA}  & \textbf{Practitioner}  & \textbf{Baseline}  \\
\hline
[1, 100]  &  0.0880 &   0.0769 &   0.0196 &  -0.0624 \\

[101, 200]  &  -0.2236 &  -0.1354  & -0.4500  & -0.3913 \\

[201, 300]  &  0.2214  &  0.0864  &  0.0148  & -0.0655 \\

[301, 400]  &  0.0533   & 0.0725 &   0.0624  & -0.0350  \\

[401, 500]  &  0.1125  &  0.1143  &  0.2973  &  0.1102 \\

[501, 600]  &  0.0231 &   0.0210 &  -0.1067 &  -0.1454  \\

[601, 700]  & 0.0678  &  0.0755 &   0.1920  &  0.0461 \\

[701, 779]  & 0.0210   & 0.0264 &  -0.0411 &  -0.1019 \\

 \bottomrule
\end{tabular}

\end{threeparttable}
\end{adjustbox}
\caption{Over-performance across models and intervals. }
    \label{Tab_over_TE}
\end{table}

Concerning the forecast accuracy, the Mean Absolute Percentage Error (MAPE) provides a measure of goodness of fit between the forecasted tracking portfolio returns and those of the benchmark.
It is defined as %
$$
MAPE=\frac{1}{T-L}\sum_{j=L+1}^T
\biggl  \lvert \frac{r_j^B-r_j^F(w)}{r_j^B} \biggr \rvert \, .
$$
According to Table \ref{T:LCPvsBAWT}, we conclude that the skew-normal hPCA strategy provides better predictions; indeed the MAPE for the hPCA (under skew-normal distribution) is  1.23\% against 1.62\% of the hPCA (under normal distribution). This is in agreement with the previous analysis on the tracking error in Table \ref{Tab_RelTE}, where the skew-normal hPCA was shown to provide smaller tracking error. In other terms, one may expect that a lower MAPE can be expected to go hand in hand with a lower tracking error.

\begin{table}[ht!]
\centering

  \begin{adjustbox}{width=.8\textwidth}
\begin{threeparttable}

\begin{tabular}{ |c|c|c|c| }
\toprule
skew-normal hPCA & normal hPCA & Practitioner & Baseline \\
\hline
1.23\% & 1.62\%   & 4.28\% & 2.99\%  \\
\hline


\end{tabular}

 \end{threeparttable}
 \end{adjustbox}
 \caption{MAPE for the considered models}
\label{T:LCPvsBAWT}
\end{table}

\subsubsection{Computational efficiency}

In the previous sections, we have examined the performance of all the approaches both in terms of  risk and reward.
We also explained the need for a solution that does not require a continual overhaul of the indexed portfolio due to exceeding a certain risk/return target as set by the investor.

In this section, we compare the computational efficiency of all the analyzed Index Tracking strategies.
Table \ref{T:LCPvsBAWTime} provides the time (in seconds) required for computing the weights of the tracking portfolios, and our hybrid PCA strategy shows the best results.
We point out that the noticeable computational efficiency is somewhat expected because, by construction, the hPCA strategy is optimizationless.
\begin{table}[ht!]
\centering
\scalebox{1.2}{
\renewcommand{\arraystretch}{1.3}
\begin{tabular}{ |c|c|c|c| }
\hline
\textbf{skew-normal hPCA} & \textbf{normal hPCA} & \textbf{Practitioner} & \textbf{Baseline} \\
\hline
0.4736 & 0.2355   & 122.75 & 2457.23  \\
\hline
\end{tabular}
}
 \caption{Average running times (in seconds) for each iteration.}
\label{T:LCPvsBAWTime}
\end{table}

\section{Conclusions} \label{Sec:Conclusion}

Active portfolio managers are not able to beat their benchmark, and those that do, cannot replicate their performances in the following years.
According to \cite{SPIVA}, in the USA, out of 703 best-performing active managers only 146 stay in the top quartile after a year, 42 after two years, 13 after three years and 2 after four years.
Furthermore, ``while the turmoil and disruption caused by the pandemic should have offered numerous opportunities for outperformance (by active managers), 57\% of domestic equity funds lagged the S\&P Composite 1500 index during the one-year period ended Dec. 31, 2020'' \citep[see][]{SPIVA2020a}.
This failure on the part of active managers has sparked much interest in index-tracking passive management.

In this paper, we have proposed an innovative statistical methodology, based on a benchmark-asset principal component factorization, for determining a tracking portfolio that replicates the performance of a benchmark by investing in a smaller number of assets.
For passive managers who need to minimize the cost of monitoring and transactions, a limited number of constituents is critical.

\noindent
We have tested and validated on real-world data the hPCA approach for normal and skew-normal returns, compared with two index tracking portfolio optimization models
used in the literature and the financial industry.
From this comparison, we have observed that the \emph{ex-post} and \emph{ex-ante} tracking errors of the hPCA skew-normal portfolios are overall lower than those of competing models, and that the hPCA skew-normal strategy also offers the best predictions.
Furthermore, the standard optimization-based model, from a practical point of view, provides frenzy tracking error expectations  making it difficult to adopt by the investment management industry.
This is a critical factor because, for portfolio managers who base the risk budget on the \emph{ex-ante} tracking error, a so erratic and misleading value may seriously disrupt the portfolio's construction.
On the other hand, the suggested skew-normal hybrid PCA is more suitable because it provides smoother changes in tracking error expectations (thus reducing the disruption in operations) and, in terms of profitability, performs better than the other strategies in turbulent markets.
Last but not least, the hPCA strategy shows remarkable computational efficiency, since such a strategy, by construction, is optimizationless.

Future research developments might be directed to investigate the impact on the risk/reward profile by changing the number of assets in a passive portfolio.

{\footnotesize
\bibliographystyle{spbasic}
\bibliography{MyBibOptim20211125}
}

\appendix

\section{Benchmark-asset principal component factorization}\label{sec:Eigen}

\subsection{Normal case} \label{app:NormalCase}
Let $\Sigma_i$ denote, for all $i \in N$, the covariance matrix of the benchmark $B$ and the asset $i$ returns, that, in a normally distributed market, is
\begin{equation}\label{Sigma1}
\Sigma_i= \left(
    \begin{array}{cc}
   \sigma_{B}^{2} & \rho_{i} \sigma_{B} \sigma_{i} \\
   \rho_{i} \sigma_{B} \sigma_{i} & \sigma_{i}^{2} \\
    \end{array}
  \right)
\end{equation}
Using the spectral theorem \citep[see, e.g.,][]{meucci2005risk}, we can write
$$\Sigma_i = E^i \Lambda^i (E^i)^{T}=\left(
    \begin{array}{cc}
   e^i_{11} & e^i_{12} \\
   e^i_{21} & e^i_{22} \\
    \end{array}
  \right)
  \left(
    \begin{array}{cc}
   \lambda^i_{1} & 0 \\
   0 & \lambda^i_{2} \\
    \end{array}
  \right)
  \left(
    \begin{array}{cc}
   e^i_{11} & e^i_{21} \\
   e^i_{12} & e^i_{22} \\
    \end{array}
  \right)$$
where $E^i$ is the matrix of the eigenvectors and $\Lambda^i$ that of the eigenvalues.
The eigenvalues can be easily obtained by imposing that $\det \left(\Sigma_i- \lambda^i I \right)=0$, namely we have $(\lambda^i)^2 - \tr(\Sigma_i) \lambda^i + \det(\Sigma_i) = 0$. Thus,
\begin{eqnarray}
 \lambda_{1}^{i}  &=& \frac{1}{2} \left[ \sigma_{B}^{2} + \sigma_{i}^{2} + \sqrt{(\sigma_{B}^{2} - \sigma_{i}^{2})^2 +4 \rho_{i}^2 \sigma_{B}^{2} \sigma_{i}^{2}} \right] \label{lam11} \\
 \lambda_{2}^{i}  &=& \frac{1}{2} \left[ \sigma_{B}^{2} + \sigma_{i}^{2} - \sqrt{(\sigma_{B}^{2} - \sigma_{i}^{2})^2 +4 \rho_{i}^2 \sigma_{B}^{2} \sigma_{i}^{2}} \right] \label{lam12}
\end{eqnarray}
To obtain the explicit expressions of the eigenvectors $e^i_1=\left(
    \begin{array}{c}
   e^i_{11} \\
   e^i_{21}\\
    \end{array}
  \right)$ and $e^i_2=\left(
    \begin{array}{c}
   e^i_{12} \\
   e^i_{22}\\
    \end{array}
  \right)$, we can solve the following equations
  $$(\Sigma_i - \lambda^i_k I) e^i_q=0  \qquad \mbox{with} \qquad q=1,2.$$
  For $q=1$, the reduced row echelon form of the matrix $(\Sigma_i - \lambda^i_1 I)$ is
  $$
  \begin{pmatrix}
  1 & \frac{2\rho_i\sigma_B\sigma_i}{\sigma_{B}^{2} - \sigma_{i}^{2} - \sqrt{(\sigma_{B}^{2} - \sigma_{i}^{2})^2 +4 \rho_{i}^2 \sigma_{B}^{2} \sigma_{i}^{2}}}\\
  0 & 0
  \end{pmatrix}
  .$$
  Therefore, we have to solve
  $$
  \begin{pmatrix}
  1 & \frac{2\rho_i\sigma_B\sigma_i}{\sigma_{B}^{2} - \sigma_{i}^{2} - \sqrt{(\sigma_{B}^{2} - \sigma_{i}^{2})^2 +4 \rho_{i}^2 \sigma_{B}^{2} \sigma_{i}^{2}}}\\
  0 & 0
  \end{pmatrix}
  \begin{pmatrix}
  e^i_{11} \\
  e^i_{21}
  \end{pmatrix}
  =
  \begin{pmatrix}
  0 \\
  0
  \end{pmatrix}.
  $$
  If we take $e^i_{21}=k$, then $e^i_{11}=\displaystyle\frac{2\rho_i\sigma_B\sigma_ik}{-(\sigma_{B}^{2} - \sigma_{i}^{2}) + \sqrt{(\sigma_{B}^{2} - \sigma_{i}^{2})^2 +4 \rho_{i}^2 \sigma_{B}^{2} \sigma_{i}^{2}}}$.
  Thus,
 \begin{equation}\label{app:eigv1}
  \begin{pmatrix}
  e^i_{11} \\
  e^i_{21}
  \end{pmatrix}
  =
  \begin{pmatrix}
  \frac{2\rho_i\sigma_B\sigma_i}{-\sigma_{B}^{2} + \sigma_{i}^{2} + \sqrt{(\sigma_{B}^{2} - \sigma_{i}^{2})^2 +4 \rho_{i}^2 \sigma_{B}^{2} \sigma_{i}^{2}}} \\
  1
  \end{pmatrix}
  k
\end{equation}
  where
 $$
  k^2=\frac{(-(\sigma_{B}^{2} - \sigma_{i}^{2}) + \sqrt{(\sigma_{B}^{2} - \sigma_{i}^{2})^2 +4 \rho_{i}^2 \sigma_{B}^{2} \sigma_{i}^{2}})^2}{4\rho^2_i\sigma^2_B\sigma^2_i+(-(\sigma_{B}^{2} - \sigma_{i}^{2}) + \sqrt{(\sigma_{B}^{2} - \sigma_{i}^{2})^2 +4 \rho_{i}^2 \sigma_{B}^{2} \sigma_{i}^{2}})^2} \, ,
$$
  since $e^2_{11}+e^2_{21}=1$.
  Analogously, the eigenvector $e^i_2$ is given by
   \begin{equation}\label{app:eigv2}
  \begin{pmatrix}
  e^i_{12} \\
  e^i_{22}
  \end{pmatrix}
  =
  \begin{pmatrix}
  \frac{2\rho_i\sigma_B\sigma_i}{-(\sigma_{B}^{2} - \sigma_{i}^{2}) - \sqrt{(\sigma_{B}^{2} - \sigma_{i}^{2})^2 +4 \rho_{i}^2 \sigma_{B}^{2} \sigma_{i}^{2}}} \\
  1
  \end{pmatrix}
  h \, ,
 \end{equation}
  where
  $$
  h^2=\frac{(-(\sigma_{B}^{2} - \sigma_{i}^{2}) - \sqrt{(\sigma_{B}^{2} - \sigma_{i}^{2})^2 +4 \rho_{i}^2 \sigma_{B}^{2} \sigma_{i}^{2}})^2}{4\rho^2_i\sigma^2_B\sigma^2_i+(-(\sigma_{B}^{2} - \sigma_{i}^{2}) - \sqrt{(\sigma_{B}^{2} - \sigma_{i}^{2})^2 +4 \rho_{i}^2 \sigma_{B}^{2} \sigma_{i}^{2}})^2}.
  $$
  Having said that, we can obtain the following principal component factorization
  $$
  \left(
    \begin{array}{c}
   R_{B} - \mu_B  \\
   R_{i} - \mu_i \\
    \end{array}
  \right)= \left(
    \begin{array}{cc}
   e^i_{11} & e^i_{12} \\
   e^i_{21} & e^i_{22} \\
    \end{array}
  \right)
  \left(
    \begin{array}{cc}
   \lambda^i_{1} & 0 \\
   0 & \lambda^i_{2} \\
    \end{array}
  \right)^{\frac{1}{2}}
  \left(
    \begin{array}{c}
   Z^i_{1}  \\
   Z^i_{2} \\
    \end{array}
  \right)
  $$
where $Z^i_{1}$ and $Z^i_{2}$ are i.i.d. standard normal random variables.
  %
  %
  %
 Hence, we have
 \begin{eqnarray}
  R_{B} &=& \mu_B + e_{11}^{i} \sqrt{\lambda_{1}^{i} } Z_{1}^{i}  + e_{12}^{i}  \sqrt{\lambda_{2}^{i} } Z_{2}^{i}  \nonumber \\
  R_{i} &=& \mu_i + e_{21}^{i} \sqrt{\lambda_{1}^{i}} Z_{1}^{i} + e_{22}^{i} \sqrt{\lambda_{2}^{i}} Z_{2}^{i} \nonumber
\end{eqnarray}

  \begin{remark}
  \noindent Let $e^i_1$ and $e^i_2$ be the eigenvectors of $\Sigma_i$ as in \eqref{app:eigv1} and \eqref{app:eigv2}, respectively.
  Then,
  \begin{eqnarray}
 e_{11}^{i} e_{21}^{i}  &=& \frac{2\rho_i\sigma_B\sigma_i}{-\sigma_{B}^{2} + \sigma_{i}^{2} + \sqrt{(\sigma_{B}^{2} - \sigma_{i}^{2})^2 +4 \rho_{i}^2 \sigma_{B}^{2} \sigma_{i}^{2}}}k^2 \nonumber \\
   &=& \frac{\rho_i\sigma_B\sigma_i(\lambda^i_1-\sigma^2_B)}{\rho^2_i\sigma^2_B\sigma^2_i+(\lambda^i_1-\sigma^2_B)^2} \label{app:e_11e_21}
\end{eqnarray}
and
  \begin{eqnarray}
 e_{12}^{i} e_{22}^{i}  &=& \frac{2\rho_i\sigma_B\sigma_i}{-\sigma_{B}^{2} + \sigma_{i}^{2} - \sqrt{(\sigma_{B}^{2} - \sigma_{i}^{2})^2 +4 \rho_{i}^2 \sigma_{B}^{2} \sigma_{i}^{2}}}h^2 \nonumber \\
   &=& \frac{\rho_i\sigma_B\sigma_i(\lambda^i_2-\sigma^2_B)}{\rho^2_i\sigma^2_B\sigma^2_i+(\lambda^i_2-\sigma^2_B)^2} \label{app:e_12e_22}
\end{eqnarray}
Furthermore, we note that in \eqref{app:e_11e_21},
\begin{equation*}
\lambda_{1}^{i}- \sigma^2_B=
  -\frac{1}{2}(\sigma_{B}^{2} - \sigma_{i}^{2}) + \frac{1}{2} \sqrt{(\sigma_{B}^{2} - \sigma_{i}^{2})^2 +4 \rho_{i}^2 \sigma_{B}^{2} \sigma_{i}^{2}} \geq 0  \qquad \forall \rho_{i} \in [-1, 1] \, ,
\end{equation*}
since $\sqrt{(\sigma_{B}^{2} - \sigma_{i}^{2})^2 +4 \rho_{i}^2 \sigma_{B}^{2} \sigma_{i}^{2}} \geq (\sigma_{B}^{2} - \sigma_{i}^{2})$;
whereas, in \eqref{app:e_12e_22}
\begin{equation*}
\lambda_{2}^{i}-\sigma^2_B=
  -\frac{1}{2} (\sigma_{B}^{2} - \sigma_{i}^{2}) - \frac{1}{2} \sqrt{(\sigma_{B}^{2} - \sigma_{i}^{2})^2 +4 \rho_{i}^2 \sigma_{B}^{2} \sigma_{i}^{2}} \leq 0 \qquad \forall \rho_{i} \in [-1, 1] \, .
\end{equation*}
%
  \end{remark}

  \subsection{Skew case}\label{Eigsn}

In case of skew-normal distributed markets, following similar arguments used in the proof of Proposition \ref{prop23}, we have that the covariance matrix $\widetilde{\Sigma}_i$
of the benchmark $B$ and the asset $i$ returns is
\begin{equation}\label{Sigma2}
\widetilde{\Sigma}_i=\biggl( 1-\frac{2\delta^2_B}{\pi}\biggr) \Sigma_i \, ,
\end{equation}
where $\Sigma_i$ is as in \eqref{Sigma1}.
Therefore, the eigenvalues of $\widetilde{\Sigma}_i$ are proportional to those of $\Sigma_i$,
namely
\begin{equation} \label{app:eigenV_Skew}
\widetilde{\lambda}^i_1=\biggl( 1-\frac{2\delta^2_B}{\pi}\biggr) \lambda^i_1, \qquad \widetilde{\lambda}^i_2=\biggl( 1-\frac{2\delta^2_B}{\pi}\biggr) \lambda^i_2,
\end{equation}
where $\lambda^i_1$ and $\lambda^i_2$ are given by \eqref{lam11}, \eqref{lam12}, respectively.
For what concerns the eigenvectors, they are the same of the normal case.


\section{Additional data information} \label{app:DataInfo}


The dataset is composed of the individual constituents belonging to the S\&P 500  (as retrieved from Yahoo Finance) and the main index alongside its industry sectors (sub-indices) taken from Bloomberg.

Figure \ref{fig1} shows the prices and log-returns of the S\&P 500 index.
\begin{figure}[!ht]
	\centering
	\begin{subfigure}
	
	\includegraphics[width=1\textwidth]{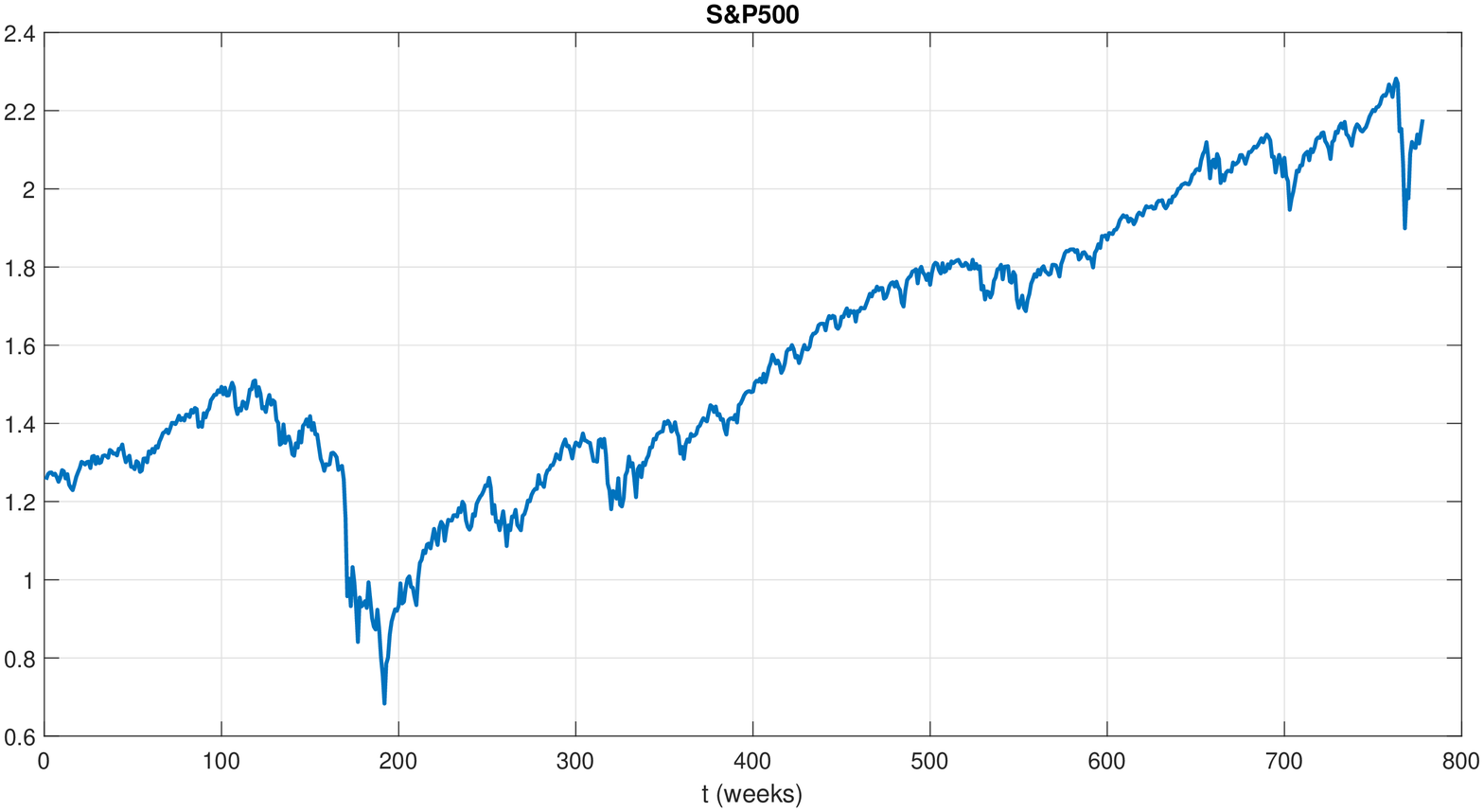}
	\end{subfigure}
     \begin{subfigure}

	\includegraphics[width=1\textwidth]{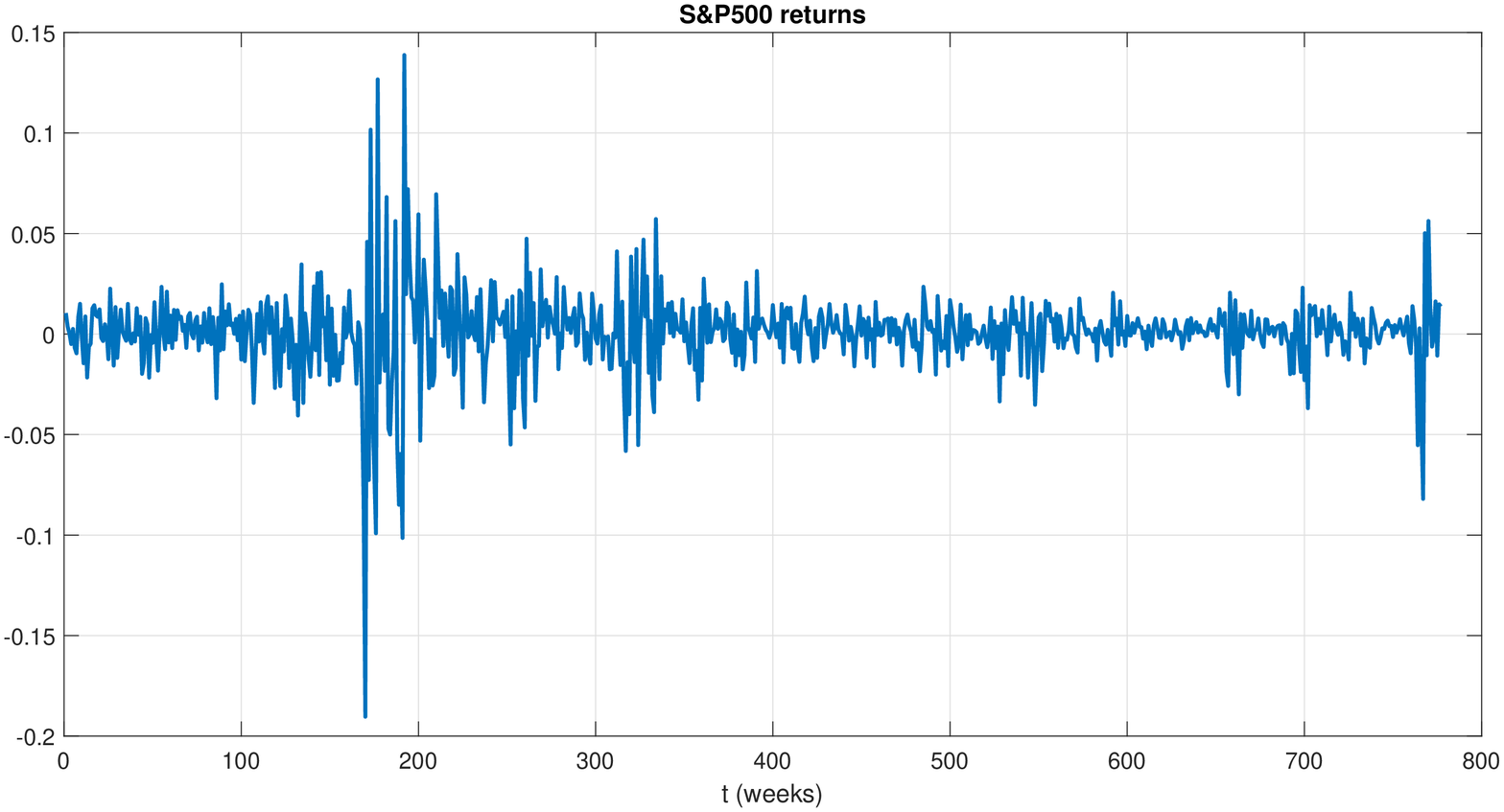}
	\end{subfigure}
	\caption{Weekly prices (top) and log-returns (bottom) of the S\&P 500 from 07 January 2005 to 29 May 2020. $T=804$ weekly observations. Source Bloomberg }
	\label{fig1}

\end{figure}
The said index includes the common stocks issued by 500 large-cap companies traded on USA stock exchanges which cover around 80\% of the equity market by capitalization.
As the index is weighted by free-float market capitalization, larger companies account more and constituents may change over time.
Indeed, between  7.01.2005 and 29.05.2020, we ended up considering 741 common stocks  (see Table \ref{T:Dataset}).
Furthermore, with reference the sub-indices mentioned in Section \ref{Sec:practitioners},  Table \ref{Tab:SP500} reports the $K=10$ sectors that make up the S\&P 500 index.

\begin{table}[!ht]
\centering
\begin{tabular}{@{}lll@{}}
\toprule
\textbf{Type} & \textbf{Bloomberg Ticker}    & \textbf{Bloomberg Name}                       \\ \midrule
Index     & SPX Index  & S\&P 500 INDEX             \\
Sub-Index & S5ENRS Index & S\&P 500 ENERGY INDEX      \\
Sub-Index & S5FINL Index & S\&P 500 FINANCIALS INDEX  \\
Sub-Index & S5INDU Index & S\&P 500 INDUSTRIALS IDX   \\
Sub-Index & S5MATR Index & S\&P 500 MATERIALS INDEX   \\
Sub-Index & S5UTIL Index & S\&P 500 UTILITIES INDEX   \\
Sub-Index & S5CONS Index & S\&P 500 CONS STAPLES IDX  \\
Sub-Index & S5TELS Index & S\&P 500 COMM SVC          \\
Sub-Index & S5COND Index & S\&P 500 CONS DISCRET IDX  \\
Sub-Index & S5HLTH Index & S\&P 500 HEALTH CARE IDX   \\
Sub-Index & S5TECH Index & S\&P 500 TECH HW \& EQP IX \\

\bottomrule
\end{tabular}
\caption{S\&P 500 and list of the sectors (sub-indices) in which the index is composed. Data from 07 January 2005 to 29 May 2020. Source: Bloomberg.}\label{Tab:SP500}
\end{table}

Table \ref{T:Dataset} reports the "signature" of the $i$-th asset for any $i$.


\begin{table}[!ht]

  \centering
  \caption{List of securities by code}  \label{T:Dataset}

  \begin{adjustbox}{width=1\textwidth}
\begin{threeparttable}

\begin{tabular}{@{}lllllllllllll@{}}
\toprule
\multicolumn{12}{c}{\textbf{ID and numbering of S\&P500 constituents}}                                                                                       &             \\ \midrule
1 = A      & 2 = AAL    & 3 = AAP    & 4 = AAPL   & 5 = ABBV    & 6 = ABC     & 7 = ABK    & 8 = ABMD   & 9 = ABS    & 10 = ABT   & 11 = ACAS  & 12 = ACN    & 13 = ADBE   \\
15 = ADM   & 16 = ADP   & 17 = ADS   & 18 = ADSK  & 19 = ADT    & 20 = AEE    & 21 = AEP   & 22 = AES   & 23 = AET   & 24 = AFL   & 25 = AGN   & 26 = AIG    & 27 = AIV    \\
29 = AJG   & 30 = AKAM  & 31 = AKS   & 32 = ALB   & 33 = ALGN   & 34 = ALK    & 35 = ALL   & 36 = ALLE  & 37 = ALTR  & 38 = ALXN  & 39 = AMAT  & 40 = AMCR   & 41 = AMD    \\
43 = AMG   & 44 = AMGN  & 45 = AMP   & 46 = AMT   & 47 = AMZN   & 48 = AN     & 49 = ANDV  & 50 = ANET  & 51 = ANF   & 52 = ANR   & 53 = ANSS  & 54 = ANTM   & 55 = AON    \\
57 = APA   & 58 = APC   & 59 = APD   & 60 = APH   & 61 = APOL   & 62 = APTV   & 63 = ARE   & 64 = ARG   & 65 = ATI   & 66 = ATO   & 67 = ATVI  & 68 = AV     & 69 = AVB    \\
71 = AVP   & 72 = AVY   & 73 = AWK   & 74 = AXP   & 75 = AYE    & 76 = AYI    & 77 = AZO   & 78 = BA    & 79 = BAC   & 80 = BAX   & 81 = BBBY  & 82 = BBY    & 83 = BCR    \\
85 = BEAM  & 86 = BEN   & 87 = BF.B  & 88 = BHF   & 89 = BIG    & 90 = BIIB   & 91 = BIO   & 92 = BJS   & 93 = BK    & 94 = BKNG  & 95 = BKR   & 96 = BLK    & 97 = BLL    \\
99 = BMS   & 100 = BMY  & 101 = BR   & 102 = BRCM & 103 = BRK.B & 104 = BS    & 105 = BSX  & 106 = BTU  & 107 = BWA  & 108 = BXLT & 109 = BXP  & 110 = C     & 111 = CA    \\
113 = CAH  & 114 = CAM  & 115 = CARR & 116 = CAT  & 117 = CB    & 118 = CBE   & 119 = CBOE & 120 = CBRE & 121 = CCE  & 122 = CCI  & 123 = CCK  & 124 = CCL   & 125 = CDNS  \\
127 = CE   & 128 = CEG  & 129 = CELG & 130 = CEPH & 131 = CERN  & 132 = CF    & 133 = CFG  & 134 = CFN  & 135 = CHD  & 136 = CHK  & 137 = CHRW & 138 = CHTR  & 139 = CI    \\
141 = CL   & 142 = CLF  & 143 = CLX  & 144 = CMA  & 145 = CMCSA & 146 = CMCSK & 147 = CME  & 148 = CMG  & 149 = CMI  & 150 = CMS  & 151 = CNC  & 152 = CNP   & 153 = CNX   \\
155 = COG  & 156 = COL  & 157 = COO  & 158 = COP  & 159 = COST  & 160 = COTY  & 161 = COV  & 162 = CPB  & 163 = CPGX & 164 = CPRI & 165 = CPRT & 166 = CPWR  & 167 = CRM   \\
169 = CSCO & 170 = CSRA & 171 = CSX  & 172 = CTAS & 173 = CTL   & 174 = CTSH  & 175 = CTVA & 176 = CTXS & 177 = CVC  & 178 = CVH  & 179 = CVS  & 180 = CVX   & 181 = CXO   \\
183 = DAL  & 184 = DD   & 185 = DE   & 186 = DELL & 187 = DF    & 188 = DFS   & 189 = DG   & 190 = DGX  & 191 = DHI  & 192 = DHR  & 193 = DIS  & 194 = DISCA & 195 = DISCK \\
197 = DJ   & 198 = DLPH & 199 = DLR  & 200 = DLTR & 201 = DNB   & 202 = DNR   & 203 = DO   & 204 = DOV  & 205 = DOW  & 206 = DPS  & 207 = DPZ  & 208 = DRE   & 209 = DRI   \\
211 = DTV  & 212 = DUK  & 213 = DV   & 214 = DVA  & 215 = DVN   & 216 = DXC   & 217 = DXCM & 218 = EA   & 219 = EBAY & 220 = ECL  & 221 = ED   & 222 = EFX   & 223 = EIX   \\
225 = EL   & 226 = EMC  & 227 = EMN  & 228 = EMR  & 229 = ENDP  & 230 = EOG   & 231 = EP   & 232 = EQIX & 233 = EQR  & 234 = EQT  & 235 = ES   & 236 = ESRX  & 237 = ESS   \\
239 = ETFC & 240 = ETN  & 241 = ETR  & 242 = EVHC & 243 = EVRG  & 244 = EW    & 245 = EXC  & 246 = EXPD & 247 = EXPE & 248 = EXR  & 249 = F    & 250 = FANG  & 251 = FAST  \\
253 = FBHS & 254 = FCX  & 255 = FDO  & 256 = FDX  & 257 = FE    & 258 = FFIV  & 259 = FHN  & 260 = FII  & 261 = FIS  & 262 = FISV & 263 = FITB & 264 = FL    & 265 = FLIR  \\
267 = FLS  & 268 = FLT  & 269 = FMC  & 270 = FNM  & 271 = FOSL  & 272 = FOX   & 273 = FOXA & 274 = FRC  & 275 = FRE  & 276 = FRT  & 277 = FRX  & 278 = FSLR  & 279 = FTI   \\
281 = FTR  & 282 = FTV  & 283 = GAS  & 284 = GD   & 285 = GDI   & 286 = GE    & 287 = GENZ & 288 = GGP  & 289 = GHC  & 290 = GILD & 291 = GIS  & 292 = GL    & 293 = GLK   \\
295 = GM   & 296 = GMCR & 297 = GME  & 298 = GNW  & 299 = GOOG  & 300 = GOOGL & 301 = GPC  & 302 = GPN  & 303 = GPS  & 304 = GR   & 305 = GRA  & 306 = GRMN  & 307 = GS    \\
309 = GWW  & 310 = HAL  & 311 = HAR  & 312 = HAS  & 313 = HBAN  & 314 = HBI   & 315 = HCA  & 316 = HCBK & 317 = HD   & 318 = HES  & 319 = HFC  & 320 = HIG   & 321 = HII   \\
323 = HNZ  & 324 = HOG  & 325 = HOLX & 326 = HON  & 327 = HOT   & 328 = HP    & 329 = HPE  & 330 = HPQ  & 331 = HRB  & 332 = HRL  & 333 = HRS  & 334 = HSIC  & 335 = HSP   \\
337 = HSY  & 338 = HUM  & 339 = HWM  & 340 = IBM  & 341 = ICE   & 342 = IDXX  & 343 = IEX  & 344 = IFF  & 345 = IGT  & 346 = ILMN & 347 = INCY & 348 = INFO  & 349 = INTC  \\
351 = IP   & 352 = IPG  & 353 = IPGP & 354 = IQV  & 355 = IR    & 356 = IRM   & 357 = ISRG & 358 = IT   & 359 = ITT  & 360 = ITW  & 361 = IVZ  & 362 = J     & 363 = JBHT  \\
365 = JCI  & 366 = JCP  & 367 = JDSU & 368 = JEC  & 369 = JEF   & 370 = JKHY  & 371 = JNJ  & 372 = JNPR & 373 = JNS  & 374 = JNY  & 375 = JOY  & 376 = JPM   & 377 = JWN   \\
379 = KEY  & 380 = KEYS & 381 = KFT  & 382 = KG   & 383 = KHC   & 384 = KIM   & 385 = KLAC & 386 = KMB  & 387 = KMI  & 388 = KMX  & 389 = KO   & 390 = KORS  & 391 = KR    \\
393 = KSE  & 394 = KSS  & 395 = KSU  & 396 = L    & 397 = LB    & 398 = LDOS  & 399 = LEG  & 400 = LEH  & 401 = LEN  & 402 = LH   & 403 = LHX  & 404 = LIFE  & 405 = LIN   \\
407 = LLL  & 408 = LLTC & 409 = LLY  & 410 = LM   & 411 = LMT   & 412 = LNC   & 413 = LNT  & 414 = LO   & 415 = LOW  & 416 = LRCX & 417 = LSI  & 418 = LUK   & 419 = LUV   \\
421 = LVS  & 422 = LW   & 423 = LXK  & 424 = LYB  & 425 = LYV   & 426 = M     & 427 = MA   & 428 = MAA  & 429 = MAC  & 430 = MAR  & 431 = MAS  & 432 = MAT   & 433 = MCD   \\
435 = MCK  & 436 = MCO  & 437 = MDLZ & 438 = MDT  & 439 = MEE   & 440 = MET   & 441 = MFE  & 442 = MGM  & 443 = MHK  & 444 = MHS  & 445 = MI   & 446 = MIL   & 447 = MJN   \\
449 = MKTX & 450 = MLM  & 451 = MMC  & 452 = MMI  & 453 = MMM   & 454 = MNK   & 455 = MNST & 456 = MO   & 457 = MOLX & 458 = MON  & 459 = MOS  & 460 = MPC   & 461 = MRK   \\
463 = MS   & 464 = MSCI & 465 = MSFT & 466 = MSI  & 467 = MTB   & 468 = MTD   & 469 = MU   & 470 = MUR  & 471 = MWW  & 472 = MXIM & 473 = MYL  & 474 = NAVI  & 475 = NBL   \\
477 = NCLH & 478 = NDAQ & 479 = NE   & 480 = NEE  & 481 = NEM   & 482 = NFLX  & 483 = NFX  & 484 = NI   & 485 = NKE  & 486 = NKTR & 487 = NLOK & 488 = NLSN  & 489 = NOC   \\
491 = NOVL & 492 = NOW  & 493 = NRG  & 494 = NSC  & 495 = NSM   & 496 = NTAP  & 497 = NTRS & 498 = NUE  & 499 = NVDA & 500 = NVLS & 501 = NVR  & 502 = NWL   & 503 = NWS   \\
505 = NYT  & 506 = NYX  & 507 = O    & 508 = ODFL & 509 = ODP   & 510 = OI    & 511 = OKE  & 512 = OMC  & 513 = ORCL & 514 = ORLY & 515 = OTIS & 516 = OXY   & 517 = PAYC  \\
519 = PBCT & 520 = PBI  & 521 = PCAR & 522 = PCG  & 523 = PCL   & 524 = PCLN  & 525 = PCP  & 526 = PCS  & 527 = PDCO & 528 = PEAK & 529 = PEG  & 530 = PEP   & 531 = PETM  \\
533 = PFG  & 534 = PG   & 535 = PGN  & 536 = PGR  & 537 = PH    & 538 = PHM   & 539 = PKG  & 540 = PKI  & 541 = PLD  & 542 = PLL  & 543 = PM   & 544 = PNC   & 545 = PNR   \\
547 = POM  & 548 = PPG  & 549 = PPL  & 550 = PRGO & 551 = PRU   & 552 = PSA   & 553 = PSX  & 554 = PTV  & 555 = PVH  & 556 = PWR  & 557 = PXD  & 558 = PYPL  & 559 = Q     \\
561 = QEP  & 562 = QRVO & 563 = QTRN & 564 = R    & 565 = RAD   & 566 = RAI   & 567 = RCL  & 568 = RDC  & 569 = RE   & 570 = REG  & 571 = REGN & 572 = RF    & 573 = RHI   \\
575 = RIG  & 576 = RJF  & 577 = RL   & 578 = RMD  & 579 = ROK   & 580 = ROL   & 581 = ROP  & 582 = ROST & 583 = RRC  & 584 = RRD  & 585 = RSG  & 586 = RSH   & 587 = RTN   \\
589 = RX   & 590 = S    & 591 = SAI  & 592 = SBAC & 593 = SBL   & 594 = SBUX  & 595 = SCG  & 596 = SCHW & 597 = SE   & 598 = SEE  & 599 = SGP  & 600 = SHLD  & 601 = SHW   \\
603 = SIG  & 604 = SII  & 605 = SIVB & 606 = SJM  & 607 = SLB   & 608 = SLE   & 609 = SLG  & 610 = SLM  & 611 = SNA  & 612 = SNDK & 613 = SNI  & 614 = SNPS  & 615 = SO    \\
617 = SPGI & 618 = SPLS & 619 = SRCL & 620 = SRE  & 621 = STE   & 622 = STI   & 623 = STJ  & 624 = STR  & 625 = STT  & 626 = STX  & 627 = STZ  & 628 = SUN   & 629 = SVU   \\
631 = SWKS & 632 = SWN  & 633 = SWY  & 634 = SYF  & 635 = SYK   & 636 = SYY   & 637 = T    & 638 = TAP  & 639 = TDC  & 640 = TDG  & 641 = TDY  & 642 = TE    & 643 = TEG   \\
645 = TER  & 646 = TFC  & 647 = TFX  & 648 = TGNA & 649 = TGT   & 650 = THC   & 651 = TIE  & 652 = TIF  & 653 = TJX  & 654 = TLAB & 655 = TMO  & 656 = TMUS  & 657 = TPR   \\
659 = TRIP & 660 = TROW & 661 = TRV  & 662 = TSCO & 663 = TSG   & 664 = TSN   & 665 = TSS  & 666 = TT   & 667 = TTWO & 668 = TWC  & 669 = TWTR & 670 = TWX   & 671 = TXN   \\
673 = TYC  & 674 = TYL  & 675 = UA   & 676 = UAA  & 677 = UAL   & 678 = UDR   & 679 = UHS  & 680 = ULTA & 681 = UNH  & 682 = UNM  & 683 = UNP  & 684 = UPS   & 685 = URBN  \\
687 = USB  & 688 = V    & 689 = VAR  & 690 = VFC  & 691 = VIAB  & 692 = VIAC  & 693 = VLO  & 694 = VMC  & 695 = VNO  & 696 = VRSK & 697 = VRSN & 698 = VRTX  & 699 = VTR   \\
701 = WAB  & 702 = WAT  & 703 = WB   & 704 = WBA  & 705 = WCG   & 706 = WDC   & 707 = WEC  & 708 = WELL & 709 = WFC  & 710 = WFM  & 711 = WFR  & 712 = WHR   & 713 = WIN   \\
715 = WM   & 716 = WMB  & 717 = WMT  & 718 = WPX  & 719 = WRB   & 720 = WRK   & 721 = WST  & 722 = WU   & 723 = WY   & 724 = WYN  & 725 = WYNN & 726 = X     & 727 = XEC   \\
729 = XL   & 730 = XLNX & 731 = XOM  & 732 = XRAY & 733 = XRX   & 734 = XTO   & 735 = XYL  & 736 = YHOO & 737 = YUM  & 738 = ZBH  & 739 = ZBRA & 740 = ZION  & 741 = ZTS
\\ \bottomrule
\end{tabular}
%
%
%

\end{threeparttable}

 \end{adjustbox}

\end{table}

%


\end{document}